\documentclass[11pt,letterpaper]{article}
\usepackage[letterpaper,margin=1in]{geometry}

\usepackage[utf8]{inputenc}
\usepackage{microtype}

\usepackage{authblk}
\usepackage{enumerate}
\usepackage{hyperref,url}
\usepackage{pgfplots}
\pgfplotsset{compat=1.14}
\usetikzlibrary{backgrounds,shapes.multipart}

\usepackage{amsmath}
\usepackage{amssymb}
\usepackage{amsthm}

\usepackage{thmtools}
\usepackage{mathtools}

\newcommand{\eqdef}{\stackrel{\normalfont\mbox{def}}{=}}

\declaretheorem{theorem}
\declaretheorem[sibling=theorem]{lemma}
\declaretheorem[sibling=theorem]{corollary}

\declaretheorem[sibling=theorem]{observation}
\declaretheorem[sibling=theorem,style=definition]{definition}
\declaretheorem[name={Open Problem}]{problem}

\let\leq\leqslant
\let\geq\geqslant

\newcommand{\polylog}{\mathrm{polylog}}
\newcommand{\Oh}[1]{\mathcal{O}{\left(#1\right)}}
\newcommand{\oh}[1]{o{\left(#1\right)}}
\newcommand{\Ot}[1]{\widetilde{\mathcal{O}}{\left(#1\right)}}
\newcommand{\Th}[1]{\Theta{\left(#1\right)}}
\newcommand{\Om}[1]{\Omega{\left(#1\right)}}
\newcommand{\inv}{\mathop{inv}}
\newcommand{\eqp}{\mathop{eqp}}
\newcommand{\invb}{\mathop{inv}}
\newcommand{\eqpb}{\mathop{eqp}}
\newcommand{\eps}{\varepsilon}
\newcommand{\pr}[1]{\mathop{Pr}[#1]}
\newcommand{\ev}[1]{\mathbb{E}[#1]}
\newcommand{\var}[1]{\mathop{Var}[#1]}
\newcommand{\Z}{\mathbb{Z}}
\newcommand{\base}{\mathop{base}}
\DeclareRobustCommand{\bbone}{\text{\usefont{U}{bbold}{m}{n}1}}

\title{Equivalences between triangle and range query problems}

\author[1]{Lech Duraj\thanks{Partially supported by the National Science Center, Poland under grant 2016/21/B/ST6/02165.}}
\author[1]{Krzysztof Kleiner}
\author[1]{Adam Polak\thanks{Partially supported by the National Science Center, Poland under grants 2017/27/N/ST6/01334 and 2018/28/T/ST6/00305.}}
\author[2]{Virginia Vassilevska Williams}

\affil[1]{
Institute of Theoretical Computer Science\\
Faculty of Mathematics and Computer Science\\
Jagiellonian University\\
\texttt{\{duraj,polak\}@tcs.uj.edu.pl}, \texttt{krzysztof.kleiner@gmail.com}
}
\affil[2]{MIT CSAIL and EECS\\\texttt{virgi@mit.edu}}

\date{}

\begin{document}

\maketitle

\begin{abstract}
We define a natural class of range query problems, and prove that all problems within this class have the same time complexity (up to polylogarithmic factors). The equivalence is very general, and even applies to online algorithms. This allows us to obtain new improved algorithms for all of the problems in the class.

We then focus on the special case of the problems when the queries are offline and the number of queries is linear. We show that our range query problems are runtime-equivalent (up to polylogarithmic factors) to counting for each edge $e$ in an $m$-edge graph the number of triangles through $e$. This natural triangle problem can be solved using the best known triangle counting algorithm, running in $\Oh{m^{2\omega/(\omega+1)}} \leq \Oh{m^{1.41}}$ time.  Moreover, if $\omega=2$, the $\Oh{m^{2\omega/(\omega+1)}}$ running time is known to be tight (within $m^{o(1)}$ factors) under the $3$SUM Hypothesis. In this case, our equivalence settles the complexity of the range query problems. Our problems constitute the first equivalence class with this peculiar running time bound.

To better understand the complexity of these problems, we also provide a deeper insight into the family of triangle problems, in particular showing black-box reductions between triangle listing and per-edge triangle detection and counting. As a byproduct of our reductions, we obtain a simple triangle listing algorithm matching the state-of-the-art for all regimes of the number of triangles. We also give some not necessarily tight, but still surprising reductions from variants of matrix products, such as the $(\min,\max)$-product.
\end{abstract}

\thispagestyle{empty}
\clearpage
\setcounter{page}{1}

\section{Introduction}
Finding, counting and listing triangles in graphs are fundamental problems with a variety of applications from classical theoretical computer science problems such as subgraph isomorphism to join query problems in databases. 

Since the 1970s~\cite{itairodeh} it has been known that triangle finding and counting can both be solved in $\Oh{n^\omega}$ time in $n$-node graphs, where $\omega<2.373$~\cite{vstoc12,LeGall14} is the exponent of square matrix multiplication. Alon, Yuster and Zwick~\cite{Alon97} improved upon this running time for sparse enough graphs by giving an $\Oh{m^{2\omega/(\omega+1)}} \leq \Oh{m^{1.41}}$ time triangle finding and counting algorithm for $m$-edge graphs. 
This is the best bound for these problems to date.

A \textsc{TriangleListing} algorithm takes as an input a graph $G$ and an integer $t$ and is required to return $t$ triangles in $G$, or all the triangles in $G$ if $G$ has fewer than $t$ triangles. The fastest known algorithms for \textsc{TriangleListing} in $m$-edge, $n$-node graphs run in either $\Ot{n^\omega+n^{3(\omega-1)/(5-\omega)} t^{2(3-\omega)/(5-\omega)}}$ time or in $\Ot{m^{2\omega/(\omega+1)}+m^{3(\omega-1)/(\omega+1)}t^{(3-\omega)/(\omega+1)}\}}$ time\footnote{We use $\Ot{\cdot}$ notation to hide polylogarithmic factors.}, depending on the graph density~\cite{Bjorklund14}. If $\omega=2$, the runtime simplifies to $\Ot{\min\{n^2 +n t^{2/3}, m^{4/3}+m t^{1/3}\}}$, and this running time has been shown to be optimal under the hypothesis that $3$SUM on $n$ integers requires $n^{2-o(1)}$ time~\cite{Patrascu10,Kopelowitz16}.

\textsc{TriangleDetection} and \textsc{TriangleListing} are important problems in graph algorithms and fine-grained complexity. Due to their simplicity, triangle problems can easily be reduced to many other problems. Fine-grained complexity has formulated hypotheses about the complexity of triangle detection and listing, and such hypotheses have been used to show lower bounds for many problems (e.g.~\cite{Patrascu10,AbboudW14}).

\textsc{TriangleDetection}\footnote{\textsc{TriangleDetection} is the problem of detecting whether a given graph has a triangle. \textsc{TriangleFinding} asks to return a triangle contained in a given graph, if one exists, and \textsc{TriangleCounting} asks to return the number of triangles in the given graph. It is not hard to see that \textsc{TriangleDetection} and \textsc{TriangleFinding} are equivalent in terms of running time. }, \textsc{TriangleCounting} and \textsc{TriangleListing} are also powerful primitives on their own. Many problems are known to be reducible to \textsc{TriangleDetection}, e.g.~$k$-{Clique} and more generally {Subgraph Isomorphism} for any fixed size pattern~\cite{nesetrilpoljak} and {Shortest Cycle}~\cite{RodittyW11}. 

\textsc{TriangleDetection} and {Boolean Matrix Multiplication} (\textsc{BMM})~\cite{focsy} were shown to be {\em equivalent} under fine-grained {\em subcubic reductions} so that an $\Oh{n^{3-\eps}}$ time algorithm for one of the problems, for any $\eps>0$, would imply an $\Oh{n^{3-\delta}}$ time algorithm for the other, for some $\delta>0$. Such reductions are tight for ``combinatorial'' algorithms.\footnote{We will not attempt to define ``combinatorial''. The notion is supposedly meant to circumvent the inefficient nature of the Strassen-like algorithms for matrix multiplication.} This relationship between the two problems also implies that many other problems are equivalent (under fine-grained subcubic reductions) to \textsc{TriangleDetection}. Some examples include {Replacement Paths} and {Shortest Cycle} in unweighted graphs~\cite{focsy}.

All these equivalences only go through for dense graphs. When the running time is measured in terms of the number of edges $m$, however, the complexities of the above problems differ a lot from the $\Oh{m^{2\omega/(\omega+1)}}$ time for \textsc{TriangleDetection}, e.g.~Shortest Cycle seems to require $mn^{1-o(1)}$ time~\cite{ancona2019,lincolnsoda18}. Prior to our work no equivalence class has been developed for triangle problems when it comes to sparse inputs. In fact, there are very few such equivalences for sparse graph problems at all: (1) There are a few equivalences for weighted graph problems whose best known running time is $\Oh{mn}$~\cite{AgarwalR18,agv}, and more interestingly, (2) a recent result~\cite{DudekG19} shows that counting $4$-cycles in $m$-edge graphs is equivalent to computing the quartet distance between two $m$-node trees, two problems with runtime $\Oh{m^{1.48}}$.

The main result of this paper is a theorem establishing an {\em equivalence class} of natural problems runtime-equivalent to the following very natural triangle counting variant, still solvable in $\Oh{m^{2\omega/(\omega+1)}}$ time by the Alon-Yuster-Zwick~\cite{Alon97} algorithm.

\begin{definition}[\textsc{EdgeTriangleCounting}] Given an undirected graph $G = (V, E)$, with $n$ nodes and $m$ edges, compute for every edge $e \in E$ the number of triangles in $G$ which contain $e$.
\end{definition}

The problems we consider are certain range query problems known to be solvable, for a linear number of queries, in $\Ot{n^{1.5}}$ time. As a byproduct of their equivalence to \textsc{EdgeTriangleCounting},
we show that they are in fact all solvable in $\Ot{n^{2\omega/(\omega+1)}} \leq \Ot{n^{1.41}}$ time.

The equivalence class is the first about problems with the bizarre complexity $\Ot{n^{2\omega/(\omega+1)}}$.
It turns out that this class has interesting relationships to other problems in fine-grained complexity such as $3$SUM, \textsc{TriangleListing}, and the $(\min,\max)$-product of matrices.

\subsection{Range query problems in our equivalence class}
\label{sec:problems}
Here we define four range query problems that are featured in our equivalence theorem. In Section~\ref{sec:results} we will define a more general range query problem that will generalize all of the problems below and will allow us to significantly extend our equivalence class. Let us define the first four.

The first problem we consider is a problem about counting the number of inversions in a set of given range queries:

\begin{definition}[\textsc{RangeInversionsQuery}]
Given an array of integers $A[1..n]$ and a sequence of ranges $[l_1,r_1],[l_2,r_2],\ldots,[l_q,r_q]$, compute for each range $[l,r]$ the quantity
\[|\{(i, j) : l \leq i < j \leq r \text{ and } A[i] > A[j]\}|.\]
\end{definition}

This is a problem commonly used to illustrate an algorithmic technique, popular under the name of {\em Mo's algorithm} in the competitive programming community, as well as referred to as the Rectilinear Steiner Minimal Arborescence technique~\cite{Kent05}. See Appendix~\ref{app:mo} for details on the technique. 

Mo's technique is very general and achieves a runtime of $\Ot{n\sqrt{q}}$ for many types of range query problems, in particular for all the range query problems in our equivalence class.
For many simple types of queries, however, faster, often (near-)linear time algorithms are known. Examples include sum (folklore), minimum~\cite{Gabow84,Bender00}, or median~\cite{Brodal11}. Counting the number of inversions seems to be one of the simplest examples for which no significant improvement over Mo's algorithm was known prior to our work.

The second problem is a variant of the first one, where we ask about two nonoverlapping ranges instead of one. An inversion is now a pair of elements from different ranges such that the left range element is larger than the right range element.

\begin{definition}[\textsc{2RangeInversionsQuery}]
Given an array of integers $A[1..n]$ and a sequence of pairs of nonoverlapping ranges $([l'_1,r'_1],[l^{''}_1,r^{''}_1]), ([l'_2,r'_2],[l^{''}_2,r^{''}_2]), \ldots, ([l'_q,r'_q],[l^{''}_q,r^{''}_q])$, with $r'_i < l^{''}_i$, compute for each pair $([l',r'],[l^{''},r^{''}])$ the quantity
\[|\{(i, j) : l' \leq i \leq r' \text{ and } l'' \leq j \leq r'' \text{ and } A[i] > A[j]\}|.\]
\end{definition}

In the third and fourth problem instead of inversions we count pairs of equal elements.

\begin{definition}[\textsc{RangeEqPairsQuery}]
Given an array of integers $A[1..n]$ and a sequence of ranges $[l_1,r_1],[l_2,r_2],\ldots,[l_q,r_q]$, compute for each range $[l,r]$ the quantity
\[|\{(i, j) : l \leq i < j \leq r \text{ and } A[i] = A[j]\}|.\]
\end{definition}

\begin{definition}[\textsc{2RangeEqPairsQuery}]
Given an array of integers $A[1..n]$ and a sequence of pairs of nonoverlapping ranges $([l'_1,r'_1],[l^{''}_1,r^{''}_1]), ([l'_2,r'_2],[l^{''}_2,r^{''}_2]), \ldots, ([l'_q,r'_q],[l^{''}_q,r^{''}_q])$, with $r'_i < l^{''}_i$, compute for each pair $([l',r'],[l^{''},r^{''}])$ the quantity
\[|\{(i, j) : l' \leq i \leq r' \text{ and } l'' \leq j \leq r'' \text{ and } A[i] = A[j]\}|.\]
\end{definition}

\subsection{Our results}\label{sec:results}
Now that all those problems have been stated, we present our equivalence theorem.

\begin{theorem}
\label{thm:eqclass}
The problems \textsc{EdgeTriangleCounting} (with input size $m$), \textsc{RangeEqPairsQuery}, \textsc{2RangeEqPairsQuery}, \textsc{RangeInversionsQuery}, \textsc{2RangeInversionsQuery} (with input sizes~$n$, restricted to offline queries and to instances with $q=\Th{n}$) all have the same time complexity in the size of their inputs, up to polylogarithmic factors.
\end{theorem}

Due to the equivalence, all these problems are solvable in $\Ot{m^{2\omega/(\omega+1)}}$ time, the time for \textsc{EdgeTriangleCounting}. This presents the first improvement over Mo's algorithm for the range query problems. Moreover, if any of these problems has a faster algorithm, then all of them have an algorithm with the same complexity. It has been open for a long time whether one can improve upon the $\Oh{m^{2\omega/(\omega+1)}}$ \textsc{TriangleDetection} runtime. Any polynomial improvement over $\Oh{n^{2\omega/(\omega+1)}}$ for the range query problems in our equivalence class would resolve this big open problem.

The $q = \Th{n}$ regime is a straightforward choice if we want to study the complexity as a function of a single parameter. The most natural parameter, the input size, equals to $\Th{n+q}$, and it is easy to see that the worst-case asymptotic complexity must be maximized for $q = \Th{n}$. Later we also analyse the complexity as a function of two parameters.

In Sections~\ref{sec:rangequery} and~\ref{sec:triangle} we prove two lemmas which together establish Theorem~\ref{thm:eqclass}.

\begin{restatable}{lemma}{RngTriEquivLem}
\label{lem:rngtri}
If \textsc{2RangeEqPairsQuery} for $q=n$ can be solved offline in $T_=(n)$ time, then \textsc{EdgeTriangleCounting} can be solved in $\Ot{T_=(m)}$ time. Conversely, if \textsc{EdgeTriangleCounting} can be solved in $T_\Delta(m)$ time, then \textsc{2RangeEqPairsQuery} for $q=n$ can be solved offline in $\Ot{T_\Delta(n)}$ time.
\end{restatable}

\begin{restatable}{lemma}{EquivClassLem}
\label{lem:eqclass}
The problems \textsc{RangeEqPairsQuery}, \textsc{2RangeEqPairsQuery}, \textsc{RangeInversionsQuery}, \textsc{2RangeInversionsQuery} all have the same time complexity, up to polylogarithmic factors. This holds even when the queries are presented online and with the complexity measured as a function of two variables, $n$ and $q$.
\end{restatable}

Since Lemma~\ref{lem:eqclass} holds for online algorithms as well, it is quite powerful.
In fact, we use it to develop improved algorithms for all of these range query problems, which work also in the online setting. In Section~\ref{sec:alg} we present an {\em online} algorithm for \textsc{RangeEqPairsQuery}, which by Lemma~\ref{lem:eqclass} also implies the same online running time for the rest of the problems. 

\begin{restatable}{theorem}{AlgThm}
\label{thm:algorithm}
\textsc{RangeEqPairsQuery} can be solved online in time
\[T(n, q) = \begin{cases}
\Ot{nq^{\frac{\omega - 1}{\omega + 1}}} &\mbox{if } q \leq n \\
\Ot{n^{\frac{2\omega - 2}{\omega + 1}}q^{\frac{2}{\omega + 1}}} &\mbox{if } q > n.\end{cases}\]
\end{restatable}
Notice that, for $q$ significantly different than $n$ our algorithm improves over the bound of $\Ot{(n+q)^{2\omega/(\omega+1)}}$ following trivially from Theorem~\ref{thm:eqclass}.

In Section~\ref{sec:alg} we compare this running time against a multivariate $3$SUM lower bound, and conclude that, assuming $\omega=2$, the bounds are tight for $q \geq n$, but a gap remains for $q<n$.

The range query problems from our equivalence class are instances of two more general problems that we will now define.
For a binary integer function $f:\Z^2\to\Z$, and an array $A$ clear from context, let us abuse the notation and write
\[f([l,r]) \eqdef \sum_{l\leq i<j\leq r} f(A[i],A[j]),
\quad \text{and} \quad f([l',r'],[l'',r'']) \eqdef
\sum_{\substack{l' \leq i \leq r' \\ l'' \leq j \leq r''}} f(A[i],A[j]).\]
This lets us define two (schemes of) problems.

\begin{definition}[\textsc{Range-$f$-PairsQuery}]
Given an array of integers $A[1..n]$ and a sequence of ranges $[l_1,r_1],[l_2,r_2],\ldots,[l_q,r_q]$, compute $f([l_i,r_i])$ for each $i\in[q]$.
\end{definition}

\begin{definition}[\textsc{2Range-$f$-PairsQuery}]
Given an array of integers $A[1..n]$ and a sequence of pairs of nonoverlapping ranges $([l'_1,r'_1],[l^{''}_1,r^{''}_1]), ([l'_2,r'_2],[l^{''}_2,r^{''}_2]), \ldots, ([l'_q,r'_q],[l^{''}_q,r^{''}_q])$, compute $f([l'_i,r'_i],[l^{''}_i,r^{''}_i])$ for each $i\in[q]$.
\end{definition}

Note that our initial four range query problems are instantiations of the above schemes for functions \[\inv(x,y) \eqdef \begin{cases}1 &\mbox{if } x > y, \\ 0 &\mbox{otherwise},\end{cases}
  \quad \text{and} \quad
  \eqp(x,y) \eqdef \begin{cases}1 &\mbox{if } x = y, \\ 0 &\mbox{otherwise}.\end{cases}\]
A natural question is: {\em What other functions yield range query problems with the same time complexity?} Labib, Uzna\'nski and Wolleb-Graf~\cite{Labib19} investigate functions equivalent to Hamming distance in the context of convolutions and matrix products. They come up with a helpful definition.

\begin{definition}[Labib, Uzna\'nski, Wolleb-Graf~\cite{Labib19}]
For integers $A,B,C$ and polynomial $P(x,y)$ we say that the function $P(x,y) \cdot \bbone[A x + B y + C > 0]$ is \emph{halfplane polynomial}.
We call a sum of halfplane polynomial functions $\sum_i P_i(x,y) \cdot \bbone[A_i x + B_i y + C_i > 0]$ \emph{piecewise polynomial}.

We say that a function is \emph{axis-orthogonal piecewise polynomial}, if it is piecewise polynomial and for every $i$, $A_i = 0$ or $B_i = 0$.
\end{definition}
Note that both $\inv$ and $\eqp$ are non-axis-orthogonal piecewise polynomial, and also many other natural functions fall within the definition. Examples include $\max(x,y)$, the $L_1$ distance $|x-y|$, more generally any odd $L_{2p+1}$ distance, the threshold function $\bbone[|x-y| < \delta]$, or the rectifier function $\max(0, x - y)$.

In Section~\ref{sec:rangequery} we integrate their techniques and vastly expand the equivalence class introduced in Theorem~\ref{thm:eqclass}.

\begin{restatable}{theorem}{fEquivClassThm}
\label{thm:feqclass}
Let $f:\Z^2\to\Z$ be any non-axis-orthogonal piecewise polynomial function of constant degree and $\mathrm{\polylog}(n)$ number of summands. For input values bounded in absolute value by $\mathrm{poly}(n)$, the problems \textsc{Range-$f$-PairsQuery} and \textsc{2Range-$f$-PairsQuery} have the same time complexity, up to polylogarithmic factors, as \textsc{2RangeEqPairsQuery}. Hence, for $q=n$, offline \textsc{Range-$f$-PairsQuery} and offline \textsc{2Range-$f$-PairsQuery} have the same time complexity, up to polylogarithmic factors, as \textsc{EdgeTriangleCounting}.
\end{restatable}

Thus \textsc{Range-$f$-PairsQuery} and \textsc{2Range-$f$-PairsQuery} (for arbitrary $f$) are equivalent to \textsc{2Range-$f$-PairsQuery} for the specific $f$ which is equality, even when the queries are presented online and with the complexity measured as a function of two variables, $n$ and $q$. This significantly extends our equivalence class.

Fine-grained complexity gives conditional lower bounds for all the problems in our expanded equivalence class. Techniques initially developed by P\v{a}tra\c{s}cu~\cite{Patrascu10} for \textsc{TriangleListing}, and further advanced by Kopelowitz, Pettie and Porat~\cite{Kopelowitz16}, can also prove that an $\Oh{m^{4/3-\eps}}$ time algorithm for $\eps>0$ for \textsc{EdgeTriangleCounting} would break the $3$SUM Hypothesis.\footnote{For a reader unassured by this hand-waving argument, let us note that the same lower bounds follow from our reductions from \textsc{TriangleListing} and \textsc{TriangleDetection}, which we shall discuss later in the paper.} Thus, we immediately obtain that under the $3$SUM Hypothesis, all problems in our equivalence class require $n^{4/3-o(1)}$ time. Thus if $\omega=2$, then we have a class whose time complexity is squarely $n^{4/3\pm o(1)}$ (under the $3$SUM Hypothesis). Moreover, improving over either the current upper bound or the current lower bound for our class by a polynomial factor would result in a significant breakthrough: Namely, an $\Oh{n^{2\omega/(\omega+1)-\eps}}$ time algorithm for $\eps>0$ would either refute the $3$SUM Hypothesis, or if the $3$SUM Hypothesis is true, then it would show that $\omega>2$. On the other hand, if one can give an $n^{4/3+\eps-o(1)}$ lower bound, then it must be that $\omega>2$.

We relate our equivalence class to two other problems of interest: \textsc{TriangleListing} and the $(\min,\max)$ matrix product. See Figure~\ref{fig:results} for an overview of the complexity landscape mapped by our results.

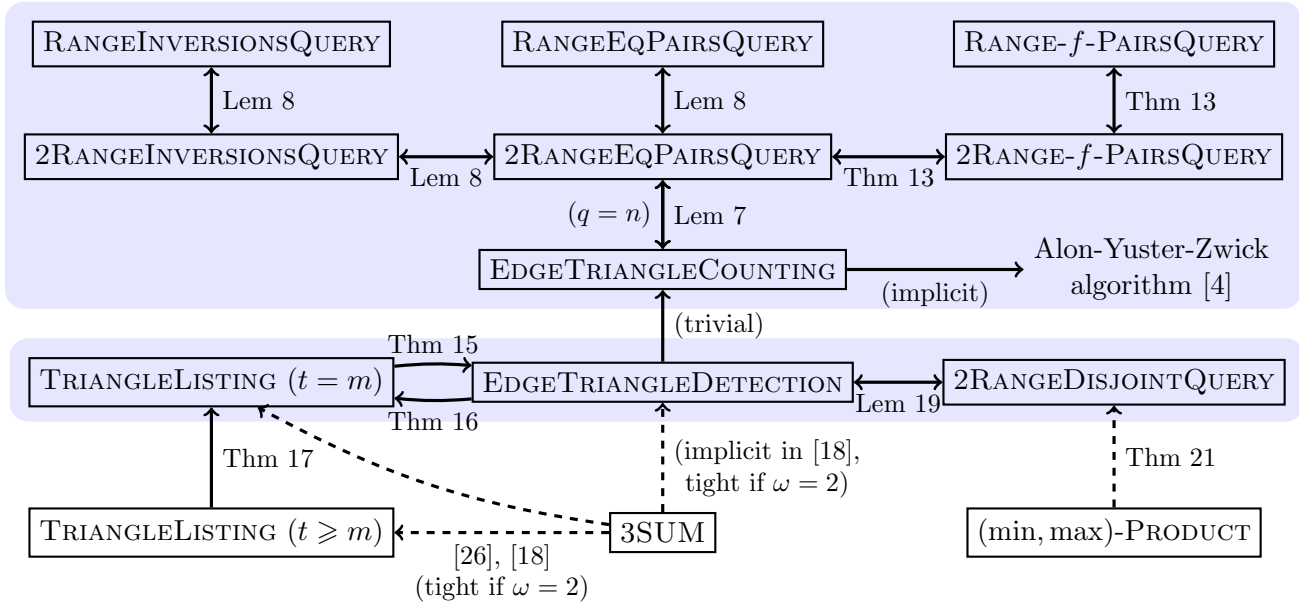
\begin{figure}
\centering
\begin{tikzpicture}
\begin{scope}[every node/.style={rectangle,thick,draw}]
  \node (2riq) at (0,6) {\textsc{2RangeInversionsQuery}};
  \node (riq)  at (0,7.5) {\textsc{RangeInversionsQuery}};
  \node (2req) at (6,6) {\textsc{2RangeEqPairsQuery}};
  \node (req)  at (6,7.5) {\textsc{RangeEqPairsQuery}};
  \node (2rfq) at (12,6) {\textsc{2Range-$f$-PairsQuery}};
  \node (rfq)  at (12,7.5) {\textsc{Range-$f$-PairsQuery}};
  \node (tcn)  at (6,4.5) {\textsc{EdgeTriangleCounting}};
  \node (tdt)  at (6,3) {\textsc{EdgeTriangleDetection}};
  \node (tlm)  at (0,3) {\textsc{TriangleListing} ($t=m$)};
  \node (tlt)  at (0,1) {\textsc{TriangleListing} ($t \geq m$)};
  \node (2rdq)  at (12,3) {\textsc{2RangeDisjointQuery}};
  \node (mmp)  at (12,1) {$(\min,\max)$-\textsc{Product}};
  \node (3sum) at (6,1) {\textsc{3SUM}};
\end{scope}
\begin{scope}[every node/.style={}]
  \node (alg) at (12.5,4.5) [align=center] {Alon-Yuster-Zwick \\ algorithm~\cite{Alon97}};
\end{scope}
\begin{scope}[every node/.style={font=\small},every edge/.style={draw=black,very thick}]
  \path [<->] (2riq) edge node[right] {Lem~\ref{lem:eqclass}} (riq);
  \path [<->] (2req) edge node[right] {Lem~\ref{lem:eqclass}} (req);
  \path [<->] (2rfq) edge node[right] {Thm~\ref{thm:feqclass}} (rfq);
  \path [<->] (2riq) edge node[below] {Lem~\ref{lem:eqclass}} (2req);
  \path [<->] (2req) edge node[below] {Thm~\ref{thm:feqclass}} (2rfq);
  \path [<->]  (tcn) edge node[right] {Lem~\ref{lem:rngtri}} (2req);
  \path [<->]  (tcn) edge node[left] {($q=n$)} (2req);
  \path [->] (tdt) edge node[right] {(trivial)} (tcn);
  \path [<-] (tlm) edge[bend right=5] node[below] {Thm~\ref{thm:detlst}} (tdt);
  \path [->] (tlm) edge[bend left=5] node[above] {Thm~\ref{thm:lstdet}} (tdt);
  \path [->] (tlt) edge node[right] {Thm~\ref{thm:lstlst}} (tlm);
  \path [->] (tcn) edge node[below] {(implicit)} (alg);
  \path [<->] (tdt) edge node[below] {Lem~\ref{lem:boolrngtri}} (2rdq);
\end{scope}
\begin{scope}[every node/.style={font=\small},every edge/.style={draw=black,very thick,dashed}]
  \path [->] (3sum) edge node[right,pos=0.4] [align=center] {(implicit in~\cite{Kopelowitz16},\\tight if $\omega=2$)} (tdt);
  \path [->] (mmp) edge node[right] {Thm~\ref{thm:minmax}} (2rdq);
  \path [->] (3sum) edge[bend left=10] (tlm);
  \path [->] (3sum) edge node[below] [align=center] {\cite{Patrascu10}, \cite{Kopelowitz16} \\ (tight if $\omega=2$)} (tlt);
\end{scope}
\begin{pgfonlayer}{background}
  \filldraw [line width=5mm,join=round,blue!10]
    (riq.north  -| 2riq.west)  rectangle (tcn.south  -| 2rfq.east)
    (tlm.north -| tlm.west) rectangle (tdt.south -| 2rdq.east);
\end{pgfonlayer}
\end{tikzpicture}
\caption{Problems, reductions and equivalence classes considered in this paper. Dashed arrows represent reductions which are not tight with respect to current fastest algorithms.}
\label{fig:results}
\end{figure}

\paragraph{Relationship to triangle listing.}
\textsc{TriangleListing} is arguably the most widely studied output-intensive triangle problem in sparse graphs~\cite{Schank05,Patrascu10,Bjorklund14}.
Quite surprisingly, in Section~\ref{sec:listing} we show that it is equivalent, in the $t=m$ regime, to the following problem, which trivially reduces to \textsc{EdgeTriangleCounting}. The reductions in that section are the first such tight reductions between triangle listing and detection problems.

\begin{definition}[\textsc{EdgeTriangleDetection}]
Given an undirected graph $G = (V, E)$, with $n$ nodes and $m$ edges, determine for every edge $e \in E$ if there exists a triangle in $G$ which contains $e$.
\end{definition}

\begin{restatable}{theorem}{LstToDetThm}
\label{thm:lstdet}
If \textsc{EdgeTriangleDetection} can be solved in $T(m)$ time, then \textsc{TriangleListing} for $t=m$ can be solved in $\Ot{T(m)}$ time.
\end{restatable}

\begin{restatable}{theorem}{DetToLstThm}
\label{thm:detlst}
If \textsc{TriangleListing} for $t=m$ can be solved in $T(m)$ time, then \textsc{EdgeTriangleDetection} can be solved in (randomized, Las Vegas) $\Ot{T(m)}$ time.
\end{restatable}

While the above reductions work in the $t=m$ regime, the next theorem lets an \textsc{EdgeTriangleDetection} algorithm be used to efficiently list an even larger number of triangles.

\begin{restatable}{theorem}{LstToLstThm}
\label{thm:lstlst}
Assume that there is an algorithm which can list up to $m$ triangles in a graph with $m$ edges in $\Ot{m^c}$ time, for a constant $c$. Then $t \geq m$ triangles can be listed in (randomized, Monte Carlo) $\Ot{m^{3c-3} t^{3-2c}}$ time.
\end{restatable}

Plugging in the $\Oh{m^{2\omega/(\omega+1)}}$ time for detection~\cite{Alon97} we match the current fastest listing algorithm~\cite{Bjorklund14}, running in $\Ot{m^{3(\omega-1)/(\omega+1)} t^{(3-\omega)/(\omega+1)}}$ time. Let us note that a fair share of intricacies in the proofs of Theorems~\ref{thm:lstdet} and~\ref{thm:lstlst} comes from the fact that they have to provide general black-box reductions. If one actually wishes to obtain a listing algorithm by chaining these reductions with the Alon-Yuster-Zwick algorithm~\cite{Alon97}, one can use the algorithm's counting ability to simplify things substantially. In particular, no randomization is needed in that case.

Last, let us define a decision variant of \textsc{2RangeEqPairsQuery}:

\begin{definition}[\textsc{2RangeDisjointQuery}]
Given an array of integers $A[1..n]$ and a sequence of pairs of nonoverlapping ranges $([l'_1,r'_1],[l^{''}_1,r^{''}_1]), ([l'_2,r'_2],[l^{''}_2,r^{''}_2]), \ldots, ([l'_q,r'_q],[l^{''}_q,r^{''}_q])$, determine for each pair $([l',r'],[l^{''},r^{''}])$ whether the sets of elements in these two ranges are disjoint, i.e. whether
\[\neg\exists_{i, j} : l' \leq i \leq r' < l'' \leq j \leq r'' \text{ and } A[i] = A[j].\]
\end{definition}

In the same way Lemma~\ref{lem:rngtri} establishes the equivalence of \textsc{2RangeEqPairsQuery} and \textsc{EdgeTriangleCounting}, we can prove \textsc{2RangeDisjointQuery} and \textsc{EdgeTriangleDetection} are equivalent:

\begin{restatable}{lemma}{BoolRngTriEquivLem}
\label{lem:boolrngtri}
If \textsc{2RangeDisjointQuery} for $q=n$ can be solved offline in $T_\perp(n)$ time, then \textsc{EdgeTriangleDetection} can be solved in $\Ot{T_\perp(m)}$ time. Conversely, if \textsc{EdgeTriangleDetection} can be solved in $T_\Delta(m)$ time, then \textsc{2RangeDisjointQuery} for $q=n$ can be solved offline in $\Ot{T_\Delta(n)}$ time.
\end{restatable}

Recall that in Section~\ref{sec:problems} we defined four counting range query problems. Out of them,  \textsc{2RangeEqPairsQuery} is the only one which remains hard in the decision variant. The other three become solvable in linear time, using range minimum query data structures~\cite{Gabow84,Bender00}.

\paragraph{Connections to matrix products.}

\textsc{2RangeEqPairsQuery} is very convenient in drawing connections of our class with matrix product problems. Let us start with a simple observation.

\begin{observation}
\label{obs:bmm}
Multiplication of two $\sqrt{n} \times \sqrt{n}$ $(0,1)$-matrices can be reduced to \textsc{2RangeEqPairsQuery} with $n$ queries in an array of length $\Oh{n}$.
\end{observation}

To see why this is true, represent each row of the first matrix and each column of the second matrix as an array of the indices in which the row/column has a 1. Then, concatenate these representations to form a single large array. Now, the value of the $(i,j)$ cell of the output matrix can be determined by asking a query about the number of pairs of equal indices in the subarrays corresponding to the row $i$ and column $j$.

The so-called BMM Hypothesis states that no ``combinatorial'' algorithm can multiply two $n \times n$ matrices, even over the Boolean semiring, in time $\Oh{n^{3-\eps}}$, for any $\eps>0$. Under this hypothesis, Mo's algorithm is optimal (up to subpolynomial factors) among ``combinatorial'' algorithms for range query problems in our class. While the notion of a ``combinatorial'' algorithm is not well defined, that observation shows that using a fast matrix multiplication algorithm -- what Alon-Yuster-Zwick algorithm~\cite{Alon97} for \textsc{EdgeTriangleCounting} does -- is necessary in order to beat the $n^{3/2}$ barrier. 

We will now relate our equivalence class to a slightly harder matrix product problem, whose complexity seems to be independent of the $3$SUM Hypothesis. The $(\min,\max)$-product of matrices $A$ and $B$ is the matrix $C$ such that $C[i][j] = \min_k \max(A[i][k], B[k][j])$. The problem of computing the $(\min,\max)$-product stems from research on the All-Pairs Bottleneck Paths problem~\cite{VWAPBP,VWAPBPj,Duan09}, and recently has been shown to be equivalent to approximating All-Pairs Shortest Paths~\cite{Bringmann19}. In Section~\ref{sec:minmax} we show the following (non-tight) reduction.

\begin{restatable}{theorem}{MinMaxThm}
\label{thm:minmax}
If \textsc{2RangeDisjointQuery} for $q=n$ can be solved offline in $T(n)$ time, then the $(\min,\max)$-product of two $n \times n$ matrices can be computed in $\Ot{T(n^2)}$ time.
\end{restatable}

The current fastest algorithm for $(\min,\max)$-product runs in $\Ot{n^{2.6865}}$ time~\cite{Duan09}. This algorithm is quite intricate and seems difficult to improve upon. Assuming this algorithm is optimal, Theorem~\ref{thm:minmax}, combined with Lemma~\ref{lem:boolrngtri} and Theorem~\ref{thm:detlst}, yields an $\Om{m^{1.3432}}$ conditional lower bound for \textsc{EdgeTriangleDetection} and \textsc{TriangleListing}. This is a slightly higher barrier than the $m^{4/3}$ one following from the $3$SUM Hypothesis, but the hardness assumption is much less understood and presumably less likely to be true.

Moreover, there is a group of equivalent matrix products~\cite{Labib19arXiv}, including dominance product~\cite{Matousek91}, equality product~\cite{cs367hw} (also later called Hamming distance product) and sparse matrix product, which can be reduced to \textsc{2RangeEqPairsQuery} (and thus \textsc{EdgeTriangleCounting}) with an even simpler argument, generalizing Observation~\ref{obs:bmm}. However these products can be computed in $\Ot{n^{2.6598}}$ time~\cite{Gold17}, slightly faster than the $(\min,\max)$-product running time, so they provide a lower bound weaker than the one based on $3$SUM, both in terms of the exponent value and perhaps credibility.

\paragraph{Preliminaries.}We assume all graphs have no isolated vertices so that the number of vertices is never asymptotically larger than the number of edges.
Throughout the paper we will omit floors and ceilings for simplicity. 

\section{Equivalence between range query problems}

\label{sec:rangequery}

In this section we present reductions between the range query problems in our equivalence class. We start with the four problems defined in Section~\ref{sec:problems}. Although their equivalence follows from a more general Theorem~\ref{thm:feqclass}, we focus on them first, so that we can highlight the main ideas behind our reductions, unobscured by technical details required for the general result.

\EquivClassLem*

\begin{proof}
 (\textsc{2RangeInversionsQuery} $\to$ \textsc{RangeInversionsQuery}):\\
 Observe that for every $a \leq b \leq c \leq d$ we have $\invb([a,b],[c,d]) = \inv([a,d]) - \inv([a,c]) - \inv([b,d]) + \inv([b,c])$.
 Thus, the answer for each pair of intervals in \textsc{2RangeInversionsQuery} can be obtained from four queries in a \textsc{RangeInversionsQuery} instance.
 
 (\textsc{RangeInversionsQuery} $\to$ \textsc{2RangeInversionsQuery}):\\
 First, note that the values $\inv([1,x])$, for all $1 \leq x \leq n$, can be precomputed in $\Oh{n \log n}$ time. To do that, note that $\inv([1,x]) - \inv([1,x-1])$ is equal to $\# \{i < x : A[i] > A[x]\}$. If we store all elements of $A[1..x-1]$ in a balanced binary search tree, the value of $\inv([1,x]) - \inv([1,x-1])$ can be found in logarithmic time. After this step, we add $A[x]$ to the tree.
 It remains to see that $\invb([1,a-1],[a,b]) = \inv([1,b]) - \inv([1,a-1]) - \inv([a,b])$, so we can find $\inv([a,b])$ with a single \textsc{2RangeInversionsQuery} $\invb([1,a-1],[a,b])$, as the other terms are precomputed.
 
 (\textsc{RangeEqPairsQuery} $\leftrightarrow$ \textsc{2RangeEqPairsQuery}): \\
 The proofs are obtained from the above ones by simply replacing $\inv$ with $\eqp$.
 
 (\textsc{2RangeEqPairsQuery} $\to$ \textsc{2RangeInversionsQuery}):\\
 To solve \textsc{2RangeEqPairsQuery} we simply use the fact that $A[i] = A[j]$ if and only if neither $A[i]>A[j]$ nor $A[i]<A[j]$. Formally, we employ a second array $A'[1..n]$ with $A'[x] = -A[x]$ and on both arrays we use the algorithm for \textsc{2RangeInversionsQuery}. Therefore, for every $a \leq b \leq c \leq d$ we can compute $\invb_A([a,b],[c,d]) = \{(i,j) \in [a,b] \times [c,d] : A[i]>A[j]\}$ as well as $\invb_{A'}([a,b],[c,d]) = \{(i,j) \in [a,b] \times [c,d] : A'[i]>A'[j]\} = \{(i,j) \in [a,b] \times [c,d] : A[i]<A[j]\}$. It is clear that $\eqpb([a,b],[c,d]) = (b-a+1)\cdot(d-c+1) - \invb_{A}([a,b],[c,d]) - \invb_{A'}([a,b],[c,d])$.
 
 (\textsc{2RangeInversionsQuery} $\to$ \textsc{2RangeEqPairsQuery}):\\
We assume w.l.o.g.~that the elements in $A$ are integers in $[1, n]$. Indeed, if they were not, we could replace each element with its position in the sorted order of all elements of $A$.

Let $k = \lceil \log n \rceil$ be the number of bits needed to represent the integers in $A$. For a $k$-bit integer $x$ and $1 \leq j \leq k$ let $v_j(x)$ denote the $j$-th bit of $x$, starting from the most significant. Also, let $p_j(x)$ denote the number obtained from $x$ by taking only the $j$ most significant bits (in other words, $p_j(x) = \lfloor x/2^{k-j} \rfloor$). We assume $p_0(x) = 0$.

We create $k$ new arrays $A_1, \ldots, A_k$, all of length $2n$. For $i = 1, 2, \ldots, n$ and $1 \leq t \leq k$ we define:

$A_t[i] = \begin{cases}
           p_{t-1}(A[i])\ &\mbox{if}\ v_t(A[i]) = 1 \\
           -\infty &\mbox{otherwise}
          \end{cases}$
          
$A_t[n+i] = \begin{cases}
           p_{t-1}(A[i])\ &\mbox{if}\ v_t(A[i]) = 0 \\
           \infty &\mbox{otherwise}
          \end{cases}$

We claim that $\invb_A([a,b],[c,d]) = \sum_{t=1}^{k} \eqpb_{A_t}([a,b],[n+c,n+d])$. This equality allows us to simulate \textsc{2RangeInversionsQuery} with $\log n$ instances of \textsc{2RangeEqPairsQuery} and thus render the proof done. To prove it, first let $i \in [a,b]$ and $j \in [c,d]$ be such that $A[i]>A[j]$. Then there is exactly one $1 \leq t \leq k$ such that $p_{t-1}(A[i]) = p_{t-1}(A[j])$, $v_t(A[i]) = 1$ and $v_t(A[j]) = 0$, which implies $A_t[i] = A_t[n+j]$. Every pair that is counted in $\invb_A([a,b],[c,d])$ then corresponds to a pair in $\sum_{t=1}^{k} \eqpb_{A_t}([a,b],[n+c,n+d])$, so $\invb_A([a,b],[c,d]) \leq \sum_{t=1}^{k} \eqpb_{A_t}([a,b],[n+c,n+d])$.

Conversely, if $A_t[i] = A_t[n+j]$ for some $a \leq i \leq b$ and $c \leq j \leq d$, then this equal element can be neither $-\infty$ nor $\infty$, as these infinities appear only in first and second half of $A_t$, respectively, and the halves are disjoint. So $p_{t-1}(A[i]) = p_{t-1}(A[j])$, $v_t(A[i]) = 1$ and $v_t(A[j]) = 0$, which means that $A[i] > A[j]$ is an inversion and $t$ is the most significant bit on which $A[i]$ and $A[j]$ differ. In particular, we cannot obtain the same pair $(i,j)$ from different $t$'s. This proves $\invb_A([a,b],[c,d]) \geq \sum_{t=1}^{k} \eqpb_{A_t}([a,b],[n+c,n+d])$, so we are done.
\end{proof}

Now we are ready to prove a more general result.

\fEquivClassThm*

The proof is similar in spirit to the proof of Lemma~\ref{lem:eqclass}, but the ad hoc reductions between the equality and inversion predicates are replaced with a general tool developed by Labib, Uzna\'nski and Wolleb-Graf~\cite{Labib19}.

\begin{theorem}[Theorem 10 in~\cite{Labib19}, rephrased]
\label{thm:polytoham}
Let $f:\Z^2\to\Z$ be any piecewise polynomial function of degree $d$ with $c$ summands. There exist integer $k=\Oh{c \cdot d \cdot \log^{d+1} U}$, constant-time computable functions $g_1,\ldots,g_k$, $h_1,\ldots,h_k$, and coefficients $\alpha_1,\ldots,\alpha_k$, such that for every $1 \leq x,y \leq U$
\[f(x,y) = \sum_{i=1}^k \alpha_i \cdot \eqp(g_i(x), h_i(y)).\]
\end{theorem}

\begin{theorem}[Theorem 11 in~\cite{Labib19}, rephrased]
\label{thm:hamtopoly}
Let $f:\Z^2\to\Z$ be any non-axis-orthogonal piecewise polynomial function of degree $d$. There exist integer $k=\Oh{d^2}$, constant-time computable functions $g_1,\ldots,g_k$, $h_1,\ldots,h_k$, and coefficients $\alpha_1,\ldots,\alpha_k$, such that for every $x, y \in \Z$
\[\eqp(x,y) = \sum_{i=1}^k \alpha_i \cdot f_i(g_i(x), h_i(y)),\]
where each $f_i$ either equals to $f$ or is a simple multiplication, i.e.~$f_i(x, y) = x \cdot y$.
\end{theorem}

Theorem~\ref{thm:feqclass} follows from the next two lemmas, which establish the equivalence between all the range query problems in our equivalence class, and Lemma~\ref{lem:rngtri}, which we prove in Section~\ref{sec:triangle}, and which relates the range query problems to \textsc{EdgeTriangleCounting}.

\begin{lemma}
Let $f:\Z^2\to\Z$ be any non-axis-orthogonal piecewise polynomial function of constant degree and $\mathrm{\polylog}(n)$ number of summands. For input values bounded in absolute value by $\mathrm{poly}(n)$, the problems \textsc{2Range-$f$-PairsQuery} and \textsc{2RangeEqPairsQuery} have the same time complexity, up to polylogarithmic factors.
\end{lemma}

\begin{proof}
(\textsc{2Range-$f$-PairsQuery} $\to$ \textsc{2RangeEqPairsQuery}):\\
We apply Theorem~\ref{thm:polytoham} to $f$, and create polylogarithmically many instances of \textsc{2RangeEqPairsQuery}, the $i$-th one with a $2n$-element array $A_i$ such that $A_i[j] = g_i(A[j])$ and $A_i[n+j] = h_i(A[j])$ for every $j\in [n]$. To finish the proof, observe that
\[f_A([l',r'],[l'',r'']) = \sum_{i=1}^k \alpha_i \cdot \eqp\nolimits_{A_i}([l',r'],[l'',r'']).\]

(\textsc{2RangeEqPairsQuery} $\to$ \textsc{2Range-$f$-PairsQuery}):\\
The argument is very similar to the reduction in the reverse direction. We use Theorem~\ref{thm:hamtopoly} and create a constant number of arrays. We then solve \textsc{2Range-$f$-PairsQuery} on those arrays with $f_i = f$, and \textsc{2Range-$\mathop{mul}$-PairsQuery} on those with $f_i = \mathop{mul}$, where $\mathop{mul}(x, y) = x \cdot y$. What remains to be shown is that \textsc{2Range-$\mathop{mul}$-PairsQuery} is computationally easy. Observe that
\[\mathop{mul}([l',r'],[l'',r'']) = \sum_{\mathclap{\substack{l' \leq i \leq r' \\ l'' \leq j \leq r''}}} A[i]A[j] = \sum_{\mathclap{l' \leq i \leq r'}} A[i] \cdot \sum_{\mathclap{l'' \leq j \leq r''}} A[j] = (S[r'+1]-S[l']) \cdot (S[r''+1]-S[l'']),\]
where $S[i] = \sum_{j < i}A[j]$, and all $S[i]$'s can be precomputed beforehand. Thus, it takes only $\Oh{n+q}$ time to solve an instance of \textsc{2Range-$\mathop{mul}$-PairsQuery}, and the total running time is necessarily dominated by solving the \textsc{2Range-$f$-PairsQuery} instances.
\end{proof}

\begin{lemma}
Let $f:\Z^2\to\Z$ be any non-axis-orthogonal piecewise polynomial function of constant degree and $\mathrm{\polylog}(n)$ number of summands. For input values bounded in absolute value by $\mathrm{poly}(n)$, the problems \textsc{Range-$f$-PairsQuery} and \textsc{2Range-$f$-PairsQuery} have the same time complexity, up to polylogarithmic factors.
\end{lemma}

\begin{proof}
(\textsc{2Range-$f$-PairsQuery} $\to$ \textsc{Range-$f$-PairsQuery}):\\
The reduction is essentially the same as the \textsc{2RangeInversionsQuery} $\to$ \textsc{RangeInversionsQuery} reduction in the proof of Lemma~\ref{lem:eqclass}. Actually, the inclusion-exclusion identity $f([a,b],[c,d]) = f([a,d]) - f([a,c]) - f([b,d]) + f([b,c])$ holds for any binary function $f$, not necessarily piecewise polynomial.

(\textsc{Range-$f$-PairsQuery} $\to$ \textsc{2Range-$f$-PairsQuery}):\\
The reduction closely follows the \textsc{RangeInversionsQuery} $\to$ \textsc{2RangeInversionsQuery} reduction in the proof of Lemma~\ref{lem:eqclass}. The inclusion-exclusion part of the argument translates verbatim, i.e.~we have $f([1,a-1],[a,b]) = f([1,b]) - f([1,a-1]) - f([a,b])$. What requires more work is to precompute the values $f([1,x])$, for all $1 \leq x \leq n$, in $\Ot{n}$ time. In order to do so, we apply Theorem~\ref{thm:polytoham} to $f$. We do a single pass over the array $A$ and, for each $i \in [k]$, we keep a multiset of already encountered values $g_i(A[x])$. During the pass, each value $f([1,x]) - f([1,x-1]) = \sum_{x'<x}f(A[x'],A[x])$ can be found in polylogarithmic time, by examining counts of $h_i(A[x])$ elements in corresponding multisets, and multiplying them by corresponding $\alpha_i$'s.
\end{proof}

\section{Equivalence with triangle counting}

\label{sec:triangle}

In this section we present reductions between \textsc{EdgeTriangleCounting} and an offline range query problem from our equivalence class. Analogous reductions establish equivalence between \textsc{EdgeTriangleDetection} and \textsc{2RangeDisjointQuery}, which we briefly discuss at the end of this section.

\RngTriEquivLem*

\begin{proof}
(\textsc{EdgeTriangleCounting} $\to$ \textsc{2RangeInversionsQuery}):\\
Given a graph, we construct an array by concatenating the lists of neighbours of all vertices. For each edge $e = (u, v)$, the number of triangles containing it, denoted by $\Delta_e$, equals  $\eqp(\mathrm{Nb}(u), \mathrm{Nb}(v))$, where $\mathrm{Nb}(a)$ denotes the interval containing the neighbour list of vertex $a$. Thus we reduce \textsc{EdgeTriangleCounting} to $m$ queries in an array of length $2m$, and can solve it in $T_=(2m)$ time.

(\textsc{2RangeEqPairsQuery} $\to$ \textsc{EdgeTriangleCounting}):\\ 
We first present a reduction producing a multigraph instance, and then explain how to eliminate parallel edges.
 
We assume $n$ is a power of 2, by appending dummy elements if necessary. Given an array of length $n$, we define the family of \textit{base intervals} as follows: for each $i = 0, 1, \ldots, \log n$, $j = 0, 1, \ldots, (n/2^i) - 1$, let $\base(i, j)$ denote the interval $[j\cdot2^i, (j+1)\cdot2^i - 1]$. This way of partitioning data is often referred to as a \textit{segment tree}. There are $2n-1$ base intervals, and their total length is $\Oh{n \log n}$.

Any interval $I$ can be split into a collection $\base_I$ of $\Oh{\log n}$ base intervals by the following recursive procedure. We start with $\base(\log n, 0) = [0, n-1]$, i.e.~the full interval. If the current interval $\base(i, j)$ is fully contained in interval $I$, we add it to the collection. Otherwise, we check if $I$ has a non-empty intersection with $\base(i-1, 2j)$ and $\base(i-1, 2j+1)$, and descend recursively into one or both of them accordingly. A conclusion that only $\Oh{\log n}$ base intervals are added to the collection (as well as that the procedure finishes in $\Oh{\log n}$ time) follows easily from the observation that, for any $i$, there can be at most two base intervals $\base(i, j)$ which have a non-empty intersection with $I$ but are not fully contained in $I$.

We create a tripartite multigraph $G=\big((U \cup V \cup W), (E_{UV} \cup E_{UW} \cup E_{VW})\big)$:
\begin{enumerate}
\item We add a vertex $u_k$ to $U$ for each value $k$ appearing in the array $A$. If a value appears multiple times, we do not create multiple copies of the vertex.
\item We add vertices $v_{i,j}$ to $V$ and $w_{i,j}$ to $W$ for each base interval $\base(i, j)$.
\item \label{step:uvuw} We add edges ($u_k$, $v_{i,j}$) to $E_{UV}$ and ($u_k$, $w_{i,j}$) to $E_{UW}$ for each value $k$ which appears in the interval $\base(i, j)$ in the array $A$. If $k$ appears multiple times, we add multiple edges, accordingly.
\item \label{step:vw} For each queried pair of intervals $(I_1, I_2)$ we compute the collections $\base_{I_1}$ and $\base_{I_2}$ and add an edge ($v_{i_1, j_1}$, $w_{i_2, j_2}$) to $E_{VW}$ for each $\base(i_1, j_1) \in \base_{I_1}$ and each $\base({i_2, j_2}) \in \base_{I_2}$.
\end{enumerate}

Since $G$ is tripartite, all of its triangles are of the form $(u_k, v_{i_1, j_1}, w_{i_2, j_2})$. Hence, for any edge $e = (v_{i_1, j_1}, w_{i_2, j_2}) \in E_{VW}$, we have $\Delta_e = \eqp(\base(i_1, j_1), \base({i_2, j_2}))$. In order to compute the results of the original queries, we sum $\Delta_e$'s for all those edges which have been induced by a given query in Step~\ref{step:vw}. This completes the reduction to a multigraph instance.

Let us analyse the size of $G$. In Step~\ref{step:uvuw}, an edge is added between each base interval and each occurrence of a value in this interval, so the total number of edges is equal to the sum of lengths of all base intervals, which is $\Oh{n \log n}$. In Step~\ref{step:vw}, for each query, both intervals are split into $\Oh{\log n}$ base intervals, and an edge is added for every pair of those base intervals. That gives $\Oh{\log^2 n}$ edges per query, $\Oh{q \log^2 n}$ edges in total. Therefore, the total size of $G$ is $\Oh{n \log n + q \log^2 n}$.

Note that multiple copies of an edge in $E_{VW}$ can be simply ignored. They occur when several queries happen to ask about the same pair of base intervals in their decompositions, but then a single copy of such an edge is sufficient to answer all those queries.

Now let us eliminate the remaining parallel edges. We represent all edge multiplicities in binary. For every $i \in [\log n]$, we create sets $E_{UV}^i$ and $E_{UW}^i$ containing those edges from $E_{UV}$ and $E_{UW}$ whose multiplicities have $1$ in the $i$-th position in their binary representation. For every pair $(i, j) \in [\log n]^2$, we create a (simple) graph $G^{i,j}=\big((U \cup V \cup W), (E_{UV}^i, E_{UW}^j, E_{VW})\big)$, and feed it to the triangle counting algorithm to compute $\Delta_e^{i, j}$ for each $e\in E_{VW}$. Finally, we obtain each $\Delta_e$ as $\sum_{i, j} 2^{i+j}\Delta_e^{i, j}$. This reduces the multigraph instance $G$ to $\log^2 n$ (simple) graph instances, each no larger than $G$. 

The whole reduction runs in time linear in the size of instances it produces. Thus, the total time to solve $\textsc{2RangeEqPairsQuery}$ is dominated by the calls to the triangle counting algorithm, and equals to $\log^2 n \cdot T_\Delta\big(\Oh{n \log n + q \log^2 n}\big)$, which is $\Ot{T_\Delta(n)}$ for $q=n$.
\end{proof}

\BoolRngTriEquivLem*

\begin{proof}[Proof sketch]
Constructions from the above proof of Lemma~\ref{lem:rngtri} work verbatim. A slight simplification is possible: Multiple parallel edges can be simply ignored, since they change only the number of triangles, not their existence. This saves a factor of $\log^2 n$.
\end{proof}

\section{Triangle detection and listing}

\label{sec:listing}

In this section we prove the equivalence between \textsc{TriangleListing} and \textsc{EdgeTriangleDetection}. Since the latter problem reduces to \textsc{EdgeTriangleCounting} trivially, this also establishes the relationship of our equivalence class to \textsc{TriangleListing}.

We start with the version of \textsc{TriangleListing} restricted to the $t=m$ regime.
\LstToDetThm*

\begin{proof}
Given a graph $\widetilde{G} = (\widetilde{V}, \widetilde{E})$, with $m=|\widetilde{E}|$ edges, we create a tripartite graph $G = (V_1 \cup V_2 \cup V_3, E)$ as follows: for every vertex $v\in \widetilde{V}$ we create three vertices $v_1\in V_1, v_2\in V_2, v_3\in V_3$ and for every edge $(u, v)\in \widetilde{E}$ we create six edges $(u_i, v_j)\in E$ for $i, j\in\{1,2,3\}, i \neq j$. Let $E = E_{1,2} \cup E_{*,3}$ where $E_{1,2}$ is the set of edges connecting $V_1$ with $V_2$ and $E_{*,3}$ is the set of edges connecting $V_1 \cup V_2$ with $V_3$. We partition $G$ into its connected components, and denote the $i$-th component by $G^i = (V^i, E^i)$, $V^i=V_1^i \cup V_2^i \cup V_3^i$, $E^i = E_{1,2}^i \cup E_{*,3}^i$. We will list $t = 6m$ triangles in $G$. Since every triangle in $\widetilde{G}$ has exactly $6$ copies in $G$, this allows us to retrieve at least $m$ unique triangles in $\widetilde{G}$.

\paragraph{Procedure.}
We iterate the whole following procedure in a loop, until the stopping condition specified in the last paragraph holds.

First, for each connected component $G^i$, we check whether $|V_3^i| = 1$. If so, we keep it unchanged. Otherwise, we are going to replace it with two new components $G^{i_1}$ and $G^{i_2}$. To do that, we arbitrarily split $V_3^i$ into two sets $V_3^{i_1}$ and $V_3^{i_2}$ of roughly equal size. We construct a new component $G^{i_1}$ as follows: We create copies of $V_1^i$, $V_2^i$ and $E_{1,2}^i$ and add them to $G^{i_1}$. We add $V_3^{i_1}$ to $G^{i_1}$. We take all those edges from $E_{*,3}^i$ which are incident to $V_3^{i_1}$, and add them to $G^{i_1}$. We construct $G^{i_2}$ analogously. We remove $G^i$ from $G$ and add $G^{i_1}$ and $G^{i_2}$ in its stead.

Note that due to these transformations the cardinality of $E_{1,2}$ might have increased by no more than a factor of 2, while the cardinality of $E_{*,3}$ hasn't changed. After all the components are examined, and possibly split, we solve \textsc{EdgeTriangleDetection} on~$G$. Then, we remove all the edges from $E_{1,2}$ which turn out not to form any triangles. If after this step it holds that $|E_{1,2}| > t$, we additionally remove arbitrary edges from $E_{1,2}$ until $|E_{1,2}| = t$. We also remove any vertices which become isolated in the process.

If there still exists a component with $|V_3^i| > 1$, we repeat the procedure. Otherwise, $|V_3^i| = 1$ for each component $i$, and therefore for each edge in $E_{1,2}$ there is exactly one vertex candidate which can form a triangle with it. Conversely, each edge still left in $E_{1,2}$ is guaranteed to form at least one triangle. We therefore list all the triangles by performing a single pass over the edges from $E_{1,2}$, and terminate.

\paragraph{Analysis.}
Observe that the only step of the algorithm in which some triangles can be lost is when we remove arbitrary edges due to the condition $|E_{1,2}| > t$. If we reach this step, however, we are guaranteed that each of the remaining $t$ edges from $E_{1,2}$ participates in at least one triangle. We conclude that if the original graph has at least $t$ triangles, each step of the algorithm preserves at least $t$ of them, while otherwise all the triangles are preserved.

The number of iterations of the procedure is $\Oh{\log{|V_3|}}$ as the cardinality of the largest $V_3^i$ decreases roughly by half with each iteration. This is $\Oh{\log{m}}$ by the assertion of there being no isolated vertices. The cardinality of $E_{*,3}$ is initially $4m$ and remains unchanged across the iterations, as each edge from $E_{*,3}^i$ is added to exactly one of the newly created components when a component gets split. Moreover, initially $|E_{1,2}| = 2m$ and then at the end of every iteration we explicitly ensure that $|E_{1,2}| \leq t = 6m$. Hence, the total number of edges at the beginning of every iteration is $\Oh{m}$, while within the iteration it can increase by no more than a factor of 2. The cost of each call to \textsc{EdgeTriangleDetection} is thus $\Oh{T(m)}$, while the rest of the procedure takes $\Oh{m}$ time per iteration. This yields the overall running time of $\Oh{T(m)\log{m}}$.
\end{proof}

\DetToLstThm*

\begin{proof}
First, observe that it is enough to solve a special case of \textsc{EdgeTriangleDetection}: assuming that $G$ is a tripartite graph $G = (V_1 \cup V_2 \cup V_3, E)$ and only detecting edges between $V_1$ and $V_2$ that are part of some triangle. Indeed, any other graph $\widetilde{G} = (\widetilde{V}, \widetilde{E})$ can be treated as in Theorem~\ref{thm:lstdet}: we create a tripartite graph $G = (V_1 \cup V_2 \cup V_3, E)$ where $V_j = \{v_j \mbox{ for each } v \in V\}$ and $E = \{(u_i, v_j): (u,v) \in \widetilde{E}, i, j\in\{1,2,3\}, i \neq j\}$. Then we call the special case algorithm for \textsc{EdgeTriangleDetection} for $G$ to detect edges between $(V_1, V_2)$. As every edge $(u,v) \in \widetilde{G}$ corresponds to an edge $(u_1,v_2) \in V_1 \times V_2$ and every triangle $(u,v,w)$ corresponds to a triangle $(u_1,v_2,w_3) \in G$, we will detect all the right edges. 

The algorithm itself is rather simple: for every $s = \log{m}, \log{m}-1, \ldots, 0$ we execute a \emph{phase} of the algorithm, which itself consists of calling  \textsc{TriangleListing} $2 \log{m}$ times. Each call is made on an induced subgraph $G^* =
(V_1 \cup V_2 \cup V_3^*, E^*)$ where the subset $V_3^* \subseteq V_3$ is created by picking every vertex from $V_3$ independently at random with probability $p = 2^{-s}$. We assume that \textsc{TriangleListing} can list up to $100m$ triangles, even if $G^*$ has less edges -- we can easily achieve that by adding dummy vertices and edges. Every edge $(v_1, v_2) \in V_1 \times V_2$ which is detected as being part of a triangle is subsequently removed from $G$. The key idea is that the phase for a given value of $s$ detects w.h.p.~all the edges with at least $2^s$ triangles, and because we promptly remove them, they cannot interfere with next phases, thus keeping the number of triangles low.

Let us formalize this idea. Recall that $\Delta_e$ denotes the number of triangles containing edge $e$. Let us define $A_s = \{e \in E\cap(V_1 \times V_2) : 2^s \leq \Delta_e < 2^{s+1}\}$ for any integer $0 \leq s \leq \log{m}$. To complete the proof, we will show the following statement: 

\emph{If all edges in $A_{s'}$ sets for $s' > s$ are detected and removed from $G$ before the $s$ phase, then, with probability $1-\frac{1}{m}$, all (previously undetected) edges in $A_s$ are detected in that phase.}

Suppose that there are indeed no more edges from $A_{s'}$ for $s' > s$. For any remaining edge $e \in \bigcup_{r\leq s} A_{r}$ the expected number of triangles in $G^*$ containing $e$ is no larger than $2^{-s}\Delta_e \leq 2^{-s} \cdot 2^{s+1} = 2$. The expected number of all triangles in $G^*$ is then no larger than $2m$, so the chance of not listing all triangles (recall that we can find up to $100m$ of them) is at most $1/50$, by the Markov inequality.
Let us now pick any edge $\gamma \in A_s$ and prove that it is likely to be detected. There are at least $2^s$ triangles with $\gamma$. For any such triangle the chance of it not appearing in $G^*$ is $1 - 2^{-s}$, and all are chosen independently, so the chance of all of them being left out is at most $(1 - 2^{-s})^{2^s} < 1/e$. From these facts we deduce that in every iteration of \textsc{TriangleListing}, $\gamma$ is detected with probability at least $1 - 1/e - 1/50 > 1/2$. With $2 \log m$ iterations we have no more than $\frac{1}{m^{2}}$ chance of missing a single edge and thus, by union bound, no more than $\frac{1}{m}$ chance of missing any edge from $A_s$, as $|A_s| \leq m$. 

This completes the proof, as the total chance that at least one phase fails is not larger than $\log m \cdot \frac{1}{m} \leq 1/2$. The running time of this algorithm is $\Oh{T(m)\log^2{m}}$. Note that we can easily detect a failure, using one extra call to \textsc{TriangleListing} with $t = 1$ at the very end of the algorithm to check whether any edge remained undetected. This makes our algorithm Las Vegas instead of Monte Carlo.
\end{proof}

\subsection{Reduction from \texorpdfstring{$t \geq m$}{t⩾m} to \texorpdfstring{$t=m$}{t=m}}

In this section we assume that there is an algorithm, denoted by BasicListingAlgorithm, which lists up to $m$ triangles in any graph $G$ with $m$ edges in time $T(m) = \Ot{m^c}$. Our goal is to use it to list up to $t \geq m$ triangles in $\Ot{m^{3c-3}t^{3-2c}}$ time.

Let $\zeta > 100$ be a big enough constant, to be determined later. If $t \leq \zeta \cdot m$, we can use BasicListingAlgorithm: we simply add $t - m \leq \zeta \cdot m$ dummy edges to $G$, which allows us to list all $t$ triangles with only a constant factor overhead in the running time. Therefore, from now on we assume that $t > \zeta \cdot m$.

Let $t^*$ denote the number of triangles in $G$. Note that we do not know $t^*$ beforehand. First we are going to deal with the case $t^* \leq t$. We present an algorithm (called InnerListingAlgorithm), which lists all triangles in a graph, provided there are at most $t$ of them. If there are more triangles in the input graph, the algorithm outputs some subset of them, with no guarantees as to the subset's size. In either case, the running time depends only on the input value $t$, and not on the actual number of triangles~$t^*$. The main idea is to invoke BasicListingAlgorithm on random subgraphs which, roughly speaking, contain no more triangles than they have edges. This is achieved by randomly coloring vertices with a carefully chosen number of colors, and iterating over tripartite graphs obtained by taking vertices of all triples of colors.

Then, we introduce MainListingAlgorithm, which  deals with the possibility of actual number of triangles being larger than $t$. It does so in a very simple way: picking vertices at random with some probability $p$ and calling InnerListingAlgorithm on the obtained subgraph. We do several iterations with various $p$, in order to make sure that the inner algorithm is likely to be invoked at least once on a subgraph of $G$ with $\Th{t}$ triangles.

\subsubsection*{Case $t^* \leq t$: InnerListingAlgorithm}

Let $r = \frac{t}{m} > \zeta$. As mentioned before, we are going to randomly assign $r$ colors to vertices, and iterate over tripartite graphs, one for every triple of colors. The expected number of triangles in such a subgraph should match its number of edges.

First, however, we need to preprocess the graph to get rid of high-degree vertices. For every vertex $v \in G$ with $\deg{v} > \frac{m}{r}$, and for every edge $(u, w) \in E$, the algorithm checks whether $(v, u, w)$ is a triangle and possibly outputs it. After this, $v$ is removed from $G$. There are no more than $2r$ such vertices, so this step works in $\Oh{r \cdot m} = \Oh{t}$ time. From now on we may assume that the degree of any vertex in $G$ is at most $\frac{m}{r}$. 

Let $q = \frac{m}{r^2} = \frac{t}{r^3} = \frac{m^3}{t^2}$. The algorithm repeats $\Oh{\log m}$ times the following pattern:

 \begin{enumerate}
  \item Assign to every $v \in V$ a random color $c_v \in \{1, 2, \ldots, r\}$.
  \item For each of the ${\binom{r}{3}}$ triples of distinct colors consider the tripartite graph $\hat{G} = (\hat{V}, \hat{E})$, which only retains the vertices of these three colors, and the edges between any distinct two of them.
   \begin{itemize}
        \item If $|\hat{E}| > \zeta \cdot q$, output nothing. (For the sake of analysing the algorithm later, we mark such a triple as \emph{failed}.)
        \item If $|\hat{E}| \leq \zeta \cdot q$, use BasicListingAlgorithm to list and output up to $\zeta \cdot q$ triangles of $\hat{G}$. (If there are exactly $\zeta \cdot q$ listed triangles, mark such a triple as \emph{vulnerable}.) 
   \end{itemize}
 \end{enumerate}
 
\begin{lemma}
\label{lem:t-listing-inner}
For any $G$ with $m$ edges and $t^*$ triangles, if $t^* \leq t$, then InnerListingAlgorithm lists all triangles in $G$ with high probability. If $t^* > t$, some subset of triangles is listed. The running time of the algorithm is always $\Ot{(\frac{t}{m})^3 \cdot T(\frac{m^3}{t^2})}$.
\end{lemma}

\begin{proof}
The running time of each of $\Oh{\log^2 m}$ iterations is clearly $\Oh{r^3 \cdot T(q)} = \Oh{(\frac{t}{m})^3 \cdot T(\frac{m^3}{t^2})}$, regardless of the number of triangles in $G$. Observe that InnerListingAlgorithm only lists triangles which appear in $G$, so the output is always a subset of all triangles.

Now we assume that there are no more than $t$ triangles in $G$ and prove that all will be listed with high probability. To do this, we pick an arbitrary triangle $\Delta = (x,y,z)$ of $G$ and prove that in any single iteration it is listed with probability at least $1/2$. If we do so, we can conclude that after $\Th{\log m}$ iterations, the probability of $\Delta$ not being listed goes down to $\frac{1}{\mathop{poly}(m)}$, with arbitrarily large $\mathop{poly}(m)$. As there are no more than $m^2$ triangles, with high probability we did not miss anything.
 
Not listing a triangle $\Delta=(x,y,z)$ can happen for one of the three reasons stated below. It is enough to prove that with the constant $\zeta$ large enough, the probability of each of them is less than $1/6$.
 \begin{enumerate}[C{a}se 1.]
  \item The colors $c_x, c_y, c_z$ are not distinct. The chance of $c_x$ being equal to $c_y$ is $1/r$, so the total chance of this bad event is no more than $\frac{3}{r} < \frac{3}{\zeta}$.
  \item The colors $c_x, c_y, c_z$ are distinct, but the triple $(c_x, c_y, c_z)$ is a \emph{failed} triple. Recall that we denote the subgraph of the chosen colors by $\hat{G} = (\hat{V}, \hat{E})$. Let us then compute the expected value of $|\hat{E}|$. For every edge $e$ with no endpoint in $\{x,y,z\}$, the probability of $e \in \hat{E}$ is $\frac{6}{r^2}$, as its endpoints must receive two distinct colors from  $\{c_x, c_y, c_z\}$. For every edge adjacent to $x$, the probability is $\frac{2}{r}$, as the other endpoint must be colored with either $c_y$ or $c_z$. There are at most $m$ edges outside of $\Delta$ and at most $\frac{3m}{r}$ adjacent to it, as every vertex has degree at most $\frac{m}{r}$; there are also three edges of $\Delta$. By linearity of expectation, the total expected number of edges in $\hat{G}$ does not exceed $3+\frac{12m}{r^2} = 12q+3 \leq 13q$. Therefore the probability of not listing $\Delta$ because of too many edges is, by the Markov inequality, no more than $\frac{13}{\zeta}$.
  \item The colors $c_x, c_y, c_z$ are distinct, but the triple $(c_x, c_y, c_z)$ is a \emph{vulnerable} triple, having more than $\zeta \cdot q$ triangles, which can lead to missing $\Delta$. This time, we calculate the expected number of triangles in $\hat{G}$. Let us split all the original triangles of $G$ into four classes: 
  \begin{itemize}
   \item The triangles disjoint with $\Delta$. Such a triangle appears in $\hat{G}$ if it receives some permutation of colors $(c_x, c_y, c_z)$ for its vertices -- the chance is $\frac{6}{r^3}$ and there are at most $t$ such triangles.
   \item The triangles with exactly one vertex common with $\Delta$. Assume that the common vertex is $x$. For any such triangle $(x,u,v)$ we know that $(u,v)$ is an edge in $G$ which is colored with $(c_y,c_z)$ or $(c_z,c_y)$. The probability of this is $\frac{2}{r^2}$ and there are $3m$ such triangles, as every one is uniquely determined by a vertex of $\Delta$ and an edge of $G$.
   \item The triangles with exactly two vertices of $\Delta$. Assume $(x,y,v)$ to be such a triangle. Then $v$ must be a neighbour of $x$ (there are at most $\deg x \leq \frac{m}{r}$ of them), and receive color $c_z$. The probability is $\frac{1}{r}$, and there are no more than $\frac{3m}{r}$ such triangles.
   \item The single triangle $\Delta$.
  \end{itemize}

  The total expected number of triangles is $\frac{6t}{r^3} + \frac{6m}{r^2} + \frac{3m}{r^2} + 1 \leq 16q$. Using Markov's inequality again, we deduce that the probability of the triple being vulnerable cannot exceed $\frac{16}{\zeta}$.
 \end{enumerate}
\end{proof}

\subsubsection*{General case: MainListingAlgorithm}
 
 The main algorithm uses InnerListingAlgorithm as a subroutine, invoking it for random subgraphs of $G$ of increasing sizes. We are going to ask InnerListingAlgorithm to list $32t$ triangles. Note that this does not increase the asymptotic running time. Moreover, the inner algorithm's internal assumption -- that the triangles to edges ratio exceeds $\zeta$ -- still holds, since we only increase the desired number of triangles and decrease the number of edges in a subgraph. As we shall see, one of the subgraphs is expected to contain between $t$ and $32t$ triangles, and thus w.h.p.~one of the calls gives us the desired answer.
 
Let $G = (V,E)$ be the input graph. Let $\mathcal{T}$ be the set of all listed triangles, initially $\mathcal{T} = \varnothing$. We consider  every $s = 0, 1, \ldots, \log m$, and for each of these values, we execute the following subroutine:
  \begin{enumerate}
   \item Choose a subset $\widetilde{V} \subseteq V$ by taking every vertex $v \in V$ with probability $p = 2^{-s}$. Let $\widetilde{G}$ be the subgraph induced by $\widetilde{V}$.
   \item Call InnerListingAlgorithm on $\widetilde{G}$ to list $32t$ triangles, add all listed triangles to $\mathcal{T}$.
   \item If $|\mathcal{T}| \geq t$, stop the algorithm.
  \end{enumerate}

\begin{lemma}
\label{lem:t-listing-outer}
With at least $3/4$ probability, there is at least one call of InnerListingAlgorithm with $\widetilde{G}$ having between $t$ and $32t$ triangles.
\end{lemma}
\begin{proof}
 Denote by $\mathcal{T}^*$ the set of all triangles in $G$. In every iteration of MainListingAlgorithm we pick vertices of $G$ with probability $p = 2^{-s}$. At least one iteration must satisfy $2t \leq p^3|\mathcal{T}^*| \leq 16t$, so let us consider this very iteration. Let $X = |\widetilde{\mathcal{T}^*}|$ be the random variable that counts the number of triangles in $\widetilde{G}$, and let us analyze the values of $\ev{X}$ and $\var{X}$. 
 
 By linearity of expectation, we have $\ev{X} = \sum_{\Delta \in \mathcal{T}^*} \pr{\Delta \in \widetilde{\mathcal{T}^*}} = p^3 |\mathcal{T}^*|$. Denote this value by $\mu$. To compute $\var{X} = \ev{X^2} - \ev{X}^2$, observe that $\ev{X^2}$ is the expected number of ordered pairs of triangles in $\widetilde{G}$, and thus it is equal to $\sum_{(\Delta_1, \Delta_2) \in \mathcal{T}^*\times\mathcal{T}^*} \pr{\Delta_1 \in \widetilde{\mathcal{T}^*} \wedge \Delta_2 \in \widetilde{\mathcal{T}^*}}$. Consider three cases, depending on the number of common vertices of $\Delta_1$ and $\Delta_2$:
 \begin{enumerate}[C{a}se 1.]
  \item For $\Delta_1$ and $\Delta_2$ disjoint, there are no more than $|\mathcal{T}^*|^2$ such pairs, and the chance for such a pair to appear is $p^6$.
  \item For $\Delta_1$ and $\Delta_2$ having one or two common vertices, we show a one-to-one mapping from such pairs to $\mathcal{T}^* \times E$. If $\Delta_1 = (a,b_1,c_1)$ and $\Delta_2 = (a,b_2,c_2)$, we map $(\Delta_1, \Delta_2)$ to $(\Delta_1, (b_2,c_2))$. If $\Delta_1 = (a,b,c_1)$ and $\Delta_2 = (a,b,c_2)$, we map $(\Delta_1, \Delta_2)$ to $(\Delta_1,(a,c_2))$. It is easy to see that we can always reconstruct $(\Delta_1, \Delta_2)$ from its assigned pair (thus confirming it is a one-to-one mapping), and that for $(\Delta_1, \Delta_2)$ to appear we need at least $\Delta_1$ and $c_2$ to be in $\widetilde{G}$. Thus there are at most $|\mathcal{T}^*| \cdot m$ such pairs, and the chance for a specific one to appear is at most $p^4$.
  \item For $\Delta_1 = \Delta_2$, there are $|\mathcal{T}^*|$ such pairs, with $p^3$ chance for each of them to appear.
 \end{enumerate} 
  We can now bound $\ev{X^2}$ from above by $p^6|\mathcal{T}^*|^2 + p^4|\mathcal{T}^*|m + p^3|\mathcal{T}^*| \leq \mu^2 +  \mu \cdot (m+1)$. Therefore, $\var{X} \leq \mu \cdot (m+1) \leq \frac{1}{16}\mu^2$, as $\mu \geq 2t \geq 16m+16$. Using Chebyshev inequality we deduce that:
  
  $$\pr{X<t \vee X>32t} \leq \pr{|X-\mu|>\frac{1}{2}\mu} \leq \frac{1}{4}.$$

This means we have no less than $3/4$ chance of calling InnerListingAlgorithm on a graph $\widetilde{G}$ with at least $t$ and at most $32t$ triangles. By Lemma \ref{lem:t-listing-inner} this will, in turn, yield $t$ triangles with high probability (in particular, with at least $3/4$ chance). This means that the probability of any of two algorithms being wrong does not exceed $1/2$.
\end{proof}

To complete the analysis, observe that MainListingAlgorithm makes $\Oh{\log^2 m}$ calls of InnerListingAlgorithm, each one of them taking $\Ot{(\frac{t}{m})^3 \cdot T(\frac{m^3}{t^2})}$ time. As stated in Lemma~\ref{lem:t-listing-outer}, one of the iterations succeeds, with high probability, in listing at least $t$ triangles. Plugging in $T(m)=m^c$ we obtain the declared result:

\LstToLstThm*

\section{Online algorithm and multivariate analysis}

\label{sec:alg}

In this section we present an online algorithm for \textsc{RangeEqPairsQuery}. Our goal is to match, for $q=\Th{n}$, the offline runtime following from Theorem~\ref{thm:eqclass}, and improve upon it for $q$ significantly different than $n$. Thanks to Lemma~\ref{lem:eqclass} and its generalization Theorem~\ref{thm:feqclass}, we automatically obtain online algorithms -- with the same running time, up to polylogarithmic factors -- for all the range query problems in our equivalence class.

For a parameter $\beta \in (0, 1)$, to be determined later, we split the input array into $n^\beta$ consecutive blocks, each consisting of $n^{1-\beta}$ consecutive elements. First, we aim to compute a matrix $B$, of size $n^\beta \times n^\beta$, such that $B[i][j]$ equals to the number of pairs of equal elements, the first element of a pair in the $i$-th block, the second element in the $j$-th block.

If a value appears in the input array at least $n^{1-\gamma}$ times (for a parameter $\gamma \in (0, 1)$ to be determined later), we call it \emph{frequent}, and otherwise we call it \emph{rare}. We will separately compute the contribution of the frequent and rare elements to the matrix $B$.

Note that there are at most $n^\gamma$ different frequent values. We construct a matrix $M$ of size $n^\beta \times n^\gamma$ such that $M[i][j]$ is the number of elements in the $i$-th block which are equal to the $j$-th frequent value. Then we use a fast matrix multiplication algorithm to compute the product $B_F = M \cdot M^T$.

Now we need to take into account the rare values. We initialize $B_R$ to the zero matrix of size $n^\beta \times n^\beta$. For every pair $(i, j)$ of equal rare elements (i.e.~$A[i]=A[j]$ and $A[i]$ is rare) we calculate indices of the blocks $b(i)$ and $b(j)$ containing $i$ and $j$, respectively, and increment $B_R[b(i)][b(j)]$ by one. We have to iterate over at most $n^{2-\gamma}$ such pairs, since each rare element can appear in at most $n^{1-\gamma}$ pairs.

Observe that $B=B_F+B_R$. Now, we compute a matrix $S$, of size $(n^\beta + 1) \times (n^\beta + 1)$, such that $S[i][j] = \sum_{i' < i, j' < j} B[i'][j']$. This is done efficiently by using the recurrence equation $S[i+1][j+1] = S[i+1][j] + S[i][j+1] - S[i][j] + B[i][j]$.

Note that, by the inclusion-exclusion principle, four lookups to $S$ can answer any query aligned to full blocks. In order to be able to handle arbitrary queries, we also store, for each value $v$ appearing in $A$, the sorted array of indices at which this value appears $I_v = \mathop{\textrm{sorted}}(\{i: A[i]=v\})$. With binary search, we can use these arrays to compute in $\Oh{\log n}$ time, for a specified index $i$, the number of elements equal to $A[i]$ in a specified range, i.e.~the number of pairs of equal elements of the form $(i, *)$ or $(*, i)$ in that range.

When a query arrives, we first identify the blocks fully contained in the range, and use $S$ to get the number of pairs of equal elements in the subrange spanned by those blocks. Note that the query asks about pairs of the form $i<j$ and $S$ includes also $i=j$ and $i>j$. To correct for this fact we subtract the length of the subrange and divide the result by $2$. Finally, we use arrays $I_v$ to include pairs with at least one of the elements in one of two ``tails'', at the beginning and at the end of the range, each of length at most the size of a block, i.e.~$n^{1-\beta}$.

The preprocessing runs in $\Oh{n^{\omega(\beta, \gamma, \beta)} + n^{2-\gamma}}$ time. Recall that $\omega(a,b,c)$ denotes the exponent in time required to multiply an $n^a \times n^b$ matrix by an $n^b \times n^c$ matrix. In particular, $\omega(a,a,a)=a\cdot\omega$. After the preprocessing, each query takes $\Oh{n^{1-\beta} \log n}$ time. In total the running time of the algorithm is $\Oh{n^{\omega(\beta, \gamma, \beta)} + n^{2-\gamma} + qn^{1-\beta} \log n}$.

For $q=\Th{n}$, it is optimal to set $\beta=\gamma=\frac{2}{\omega+1}$, and we get the time complexity $\Ot{n^{\frac{2\omega}{\omega+1}}}$, which matches the upper bound following from the reduction to \textsc{EdgeTriangleCounting} (Theorem~\ref{thm:eqclass}) and the current fastest algorithm for that problem~\cite{Alon97}.

When $n$ and $q$ differ significantly from each other, optimizing the parameters $\beta$ and~$\gamma$ leads to a multiplication of rectangular matrices. We use a naive bound $\omega(a,b,c) \leq a+b+c-(3-\omega)\cdot\min(a,b,c)$, which follows from a block-splitting argument. It is possible to obtain better bounds~\cite{LeGall18}, however they require numerical analysis for each particular set of values $a, b, c$, and thus do not yield any meaningful closed-form formula for the running time of our algorithm. Moreover, if $\omega=2$, the above naive bound turns out to be tight. Let $\alpha$ denote $\frac{\log q}{\log n}$. We optimize by setting $\beta = \frac{2\alpha}{\omega+1}$ when $q \leq n$, $\beta = \frac{3 - \alpha + \omega \cdot (\alpha + 1)}{\omega+1}$ when $q > n$, and in both cases $\gamma = 1 + \beta - \alpha$, which gives the following running times:

\begin{figure}
\centering
\begin{tikzpicture}
\begin{axis}[
  xlabel={Number of queries $q$},
  ylabel={Running time $T(n,q)$},
  xmin=0, xmax=2, ymin=1, ymax=2, xtick={0,0.5,1,1.5,2}, ytick={1,1.33,1.41,2},
  xticklabels={$1$, $n^{1/2}$, $n$, $n^{3/2}$, $n^2$},
  yticklabels={$n$, $n^{4/3}$, $n^{1.41}$, $n^2$},
  legend pos=north west, legend cell align=left,
  ymajorgrids=true, xmajorgrids=true, grid style=dashed
]
\addplot[color=red,dashed] coordinates {(0,1) (1,1.41) (2,2)}; 
\addplot[color=green,dashed] coordinates {(0,1) (1,1.33) (2,2)};
\addplot[color=blue,dashed] coordinates {(0.5,1) (1,1.33) (2,2)};
\legend{{Upper bound, $\omega=2.373$}, {Upper bound, $\omega=2$}, {Lower bound}}
\end{axis}
\end{tikzpicture}
\caption{Time complexity of \textsc{RangeEqPairsQuery}.}
\label{fig:bounds}
\end{figure}
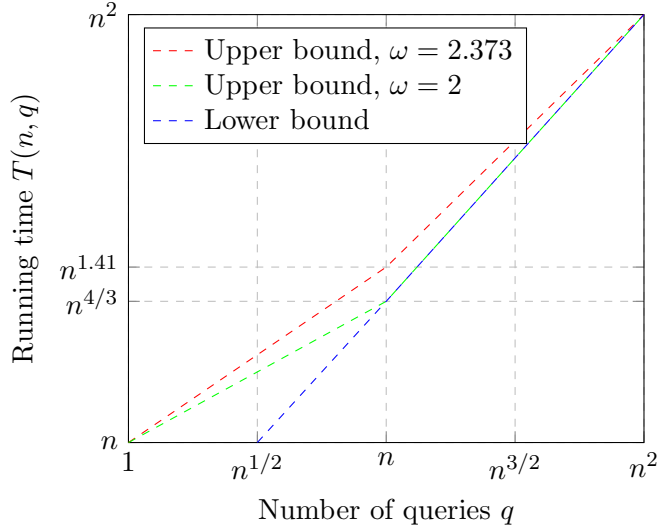

\AlgThm*

In the above description we assume implicitly that we know the number of queries $q$ in advance and can adjust parameters $\beta$ and $\gamma$ accordingly. Without this optimistic assumption we start hypothesizing $q = 1$, and whenever the actual number of queries exceeds the present guess of $q$, we multiply it by $2$, update the parameters and rerun the preprocessing. Doing so we lose only a constant multiplicative factor.

Let us compare this running time against a lower bound that follows from the set disjointness framework of Kopelowitz, Pettie and Porat~\cite{Kopelowitz16}.

\begin{theorem}[Kopelowitz, Pettie, Porat~\cite{Kopelowitz16} (rephrased)]
Unless the $3$SUM Hypothesis fails, for any constants $0\leq\lambda<1$, $\eps > 0$, there is no $\Oh{N^{2-\eps}}$ time algorithm that determines disjointness for each of $\Th{N^{1+\lambda}}$ pairs of sets from a family of $\Th{N}$ sets of size $\Th{N^{1-\lambda}}$ each.
\end{theorem}

Note that the above set disjointness problem easily reduces to \textsc{2RangeEqPairsQuery} with $\Th{N^{1+\lambda}}$ queries in an array of length $\Th{N \cdot N^{1-\lambda}}$. Therefore we have $T(N^{2-\lambda}, N^{1+\lambda}) = \Om{N^{2-\eps}}$, for any $0\leq\lambda<1$, which finally yields the following corollary.

\begin{corollary}
Unless the $3$SUM Hypothesis fails, there is no $\Oh{(nq)^{2/3-\eps}}$ time algorithm for \textsc{RangeEqPairsQuery} for any constant $\eps > 0$. This even holds restricted to instances with $q=n^{\alpha \pm \oh{1}}$, for arbitrarily chosen $\frac{1}{2} \leq \alpha < 2$.
\end{corollary}

See Figure~\ref{fig:bounds} for a visual comparison of these bounds.

\section{Reduction from \texorpdfstring{$(\min,\max)$}{(min,max)}-product}

\label{sec:minmax}

\MinMaxThm*

\begin{proof}
Let $A$ and $B$ be the input matrices. For every $i=1,2,\ldots,n$ let $\mathcal{A}_i$ denote a~permutation of the column indices $\{1,2,\ldots,n\}$ sorted by their corresponding values in the $i$-th row of $A$, that is $A[i][\mathcal{A}_i[k]] \leq A[i][\mathcal{A}_i[k+1]]$ for every $k$. Analogously, for every $j=1,2,\ldots,n$ let $\mathcal{B}_j$ denote a permutation of the row indices of the $j$-th column of $B$ sorted in nondecreasing order of the entries. We concatenate all $\mathcal{A}_i$'s and all $\mathcal{B}_j$'s together to form a single array $\mathcal{T}$.

The key idea behind our reduction is the following simple equivalence.
\[
C[i][j] = \min_k \max(A[i][k], B[k][j]) \leq x
\quad\Longleftrightarrow\quad
\{k : A[i][k] \leq x \} \cap \{k: B[k][j] \leq x \} \neq \emptyset
\]
Observe that the two sets on the right-hand side are sets of elements of a prefix of $\mathcal{A}_i$ and a prefix of $\mathcal{B}_j$. Therefore, we can learn if a particular $C[i][j]$ is below or above a given threshold by asking a single query whether two ranges in $\mathcal{T}$ have disjoint sets of elements. The lengths of the relevant prefixes can be computed just before with a simple binary search in $\Oh{\log n}$ time. With $n^2$ queries we can learn this for the entire matrix $C$, and nothing forbids us from specifying a different threshold for each matrix cell.

We can assume that the entries in $A$ and $B$, and therefore also in $C$, are integers in $[1,2n^2]$. Indeed, if they were not, we could replace each entry with its position in the overall sorted order of all entries in $A$ or $B$. Hence, it takes $\Oh{\log n}$ steps of a parallel binary search -- each consisting of solving an (offline) instance of \textsc{2RangeDisjointQuery} with $n^2$ queries in an array of $n^2$ elements -- to compute the matrix $C$. 
\end{proof}

\section{Open problems}

A notable absence in our triangle-related problems landscape is \textsc{TriangleCounting} -- finding the total number of \emph{all} triangles in an $m$-edge graph.
The fastest known algorithm is the one provided by Alon, Yuster and Zwick~\cite{Alon97}, working in $\Oh{m^{2\omega/(\omega+1)}}$ time, which anyway solves \textsc{EdgeTriangleCounting} as a byproduct.
While a possible faster algorithm for \textsc{TriangleCounting} has not been (conditionally) ruled out, we believe that there may exist a reduction proving the equivalence of \textsc{EdgeTriangleCounting} and \textsc{TriangleCounting}.
\begin{problem}
Are \textsc{TriangleCounting} and \textsc{EdgeTriangleCounting} equivalent? If not, can we solve \textsc{TriangleCounting} faster than $\Oh{m^{2\omega/(\omega+1)}}$ time?
\end{problem}

The equivalent triangle and range query problems considered in this paper have a (conditional) lower bound of $m^{4/3-\oh{1}}$ and an upper bound of $\Oh{m^{2\omega/(\omega+1)}} \leq \Oh{m^{1.41}}$. Bridging this gap is a long-standing open question.
Obviously, if $\omega=2$, the question would be resolved. Suppose, however, that $\omega>2$ but a faster (e.g.~$\Ot{m^{4/3}}$ time) algorithm is found for \textsc{EdgeTriangleCounting} and thus, by Theorems~\ref{thm:eqclass} and~\ref{thm:feqclass}, for all our offline range query problems. Would it also imply faster online algorithms? In other words, can we use an offline range query algorithm as a black-box to find an online algorithm of the same complexity? As of now, we can neither provide such an equivalence, nor prove a higher conditional lower bound for online variants.

\begin{problem}
Are the offline and online variants of our range query problems equivalent?
\end{problem}

Finally, as can be seen in Figure~\ref{fig:bounds}, our lower bound and upper bound for range query problems do not match when the number of queries is sublinear in $n$, even if $\omega=2$.
For example, for $q = n^{1/2\pm\oh{1}}$ the lower bound becomes trivial $\Om{n}$, and the upper bound is $\Ot{n^{(3\omega+1)/(2\omega+2)}}$, which is at best $\Ot{n^{7/6}}$, if $\omega=2$.

\begin{problem}
What is the complexity of the range query problems for $q = n^{\alpha\pm\oh{1}}$ and $\alpha \in (0,1)$?
\end{problem}

\appendix

\section{Mo's algorithm}

\label{app:mo}

Mo's algorithm became a standard tool in the competitive programming community, but it seems virtually nonexistent in the theoretical computer science literature. For the sake of completeness, in this section we present the original offline algorithm, and propose how to generalize it for online problems by using persistent data structures.

\subsection{Offline algorithm}

For common range query problems the following observation holds: it is possible -- usually with a help of a simple data structure, e.g.~a binary search tree -- to quickly transform the answer for the current range into the answer for a range one element shorter or longer at either end, and then update the data structure, so that successive such transformations can be applied again and again. If that operation can be performed in polylogarithmic time, then a straightforward algorithm can answer $q$ queries in an array of length $n$ in $\Ot{n^2 + q}$ time, by precomputing all possible queries. To simplify further analysis we will assume $q \leq n^2$, otherwise the above algorithm is optimal. Another naive approach is to do no preprocessing and compute the answer to each query from scratch, in time $\Ot{nq}$. Mo's technique provides a simple framework reducing the runtime down to $\Ot{n\sqrt{q}}$.

Recall that we denote the $i$-th query by $[l_i, r_i]$, for $i=1,2,\ldots,q$. Observe that we can answer all the queries, one after the other, by performing $\sum_i \big(|l_i-l_{i-1}| + |r_i-r_{i-1}|\big)$ operations on the underlying data structure, each taking polylogarithmic time. Mo's central idea is to make this sum small by leveraging the fact that the algorithm works offline, i.e.~it knows all the queries in advance, and thus can handle them in a favorable order. Therefore, for a parameter $B$ to be determined later, it sorts the queries by $\lfloor l_i / B \rfloor$ and in case of a tie by $r_i$. To avoid double indexing, in the analysis below we work with this rearranged order.

Observe that $\sum_i |l_i-l_{i-1}| \leq \Oh{n + B\cdot q}$, because $|l_i - l_{i-1}|$ = $(l_i - l_{i-1}) + 2\cdot\max(l_{i-1} - l_i, 0)$, and $\sum_i (l_i-l_{i-1}) = l_q - l_1 \leq n$, while the ordering guarantees $l_{i-1} - l_i \leq B$. On the other hand, $\sum_i |r_i-r_{i-1}| \leq \Oh{(n/B)\cdot n}$, as there are at most $n/B$ different values of $\lfloor l_i / B \rfloor$ and the queries are sorted by increasing values of $r_i$'s within each value of $\lfloor l_i / B \rfloor$. This gives the total running time of $\Ot{B\cdot q+(n/B)\cdot n + n + q}$. We optimize it by setting $B=n/\sqrt{q}$ (thanks to the $q \leq n^2$ assumption we are guaranteed that $B\geq 1$), and get the desired $\Ot{n\sqrt{q}}$ time bound.

\subsection{Online algorithm}

Mo's algorithm can only be used offline, i.e.~when all the queries are known beforehand. There are other general methods for solving range query problems, which do work in the online setting, most notably a \emph{square root decomposition} method, i.e.~dividing the array into $\sqrt{n}$ consecutive blocks. However, for some problems, e.g.~calculating the number of distinct elements within a range, these methods do not seem to apply, while Mo's algorithm gives an immediate (albeit offline-only) solution. Fortunately, for usual range query problems Mo's algorithm uses common data structures, ones which are known to have \emph{fully persistent} equivalents -- a circumstance we shall exploit to design an online algorithm. Our solution matches the time complexity of the original Mo's algorithm, but it requires $\Ot{n\sqrt{q}}$ space while the offline version uses only linear space.

A data structure is called \textit{persistent} if it supports access to its older versions, called \emph{snapshots}. A persistent data structure is called \textit{fully persistent} if its snapshots can be modified, and therefore it allows working with a branched history. Many common data structures have been shown to have fully persistent equivalents with desirable properties, i.e.~polylogarithmic access and update times and constant space utilization per update~\cite{Driscoll89,Okasaki98}.

Suppose that it is possible to transform the answer for a given range into the answer for a range one element longer at either end\footnote{Note that, unlike in standard Mo's algorithm, we do not require computing the answer for a range one element shorter.} with a fully persistent data structure $S$ with polylogarithmic access and update times. In the preprocessing phase we choose a parameter $B$ and for each $j = 0,1,\ldots,\lfloor n/B\rfloor$ we incrementally compute (and store) the results for hypothetical queries $[B\cdot j, k]$ for each $k \in \{B\cdot j,\ldots,n\}$. That is, we keep one instance of $S$ per each $j$, and we take a snapshot of it for each $k$. By virtue of the properties of fully persistent data structures, this takes $\Ot{(n/B) \cdot n}$ time and space.

When a query $[l, r]$ arrives, we take the precomputed result and the related data structure snapshot for the hypothetical query $\big[B\cdot\lceil l/B \rceil, r\big]$ (unless $B\cdot\lceil l/B \rceil > r$, in which case we create an ad hoc instance of the data structure for the one-element $[r,r]$ range). We then extend the result in $\Ot{B}$ time by adding elements at the front of the range, one element at a time, until we reach $[l, r]$.

If we have an a priori bound on the number of queries, we optimize by setting $B=n/\sqrt{q}$, just like for the offline Mo's algorithm, and obtain an $\Ot{n\sqrt{q}}$ running time. Without such knowledge, we proceed analogously as in the algorithm of Theorem~\ref{thm:algorithm}: We start hypothesizing $q = 1$, and whenever the actual number of queries exceeds the present guess of $q$, we multiply it by $2$, adjust $B$ and rerun the preprocessing. Doing so we only lose a constant multiplicative factor.

\bibliography{paper}

\begin{thebibliography}{10}

\bibitem{agv}
Amir Abboud, Fabrizio Grandoni, and Virginia {Vassilevska Williams}.
\newblock Subcubic equivalences between graph centrality problems, {APSP} and
  diameter.
\newblock In {\em Proceedings of the Twenty-Sixth Annual {ACM-SIAM} Symposium
  on Discrete Algorithms, {SODA} 2015, San Diego, CA, USA, January 4-6, 2015},
  pages 1681--1697, 2015.

\bibitem{AbboudW14}
Amir Abboud and Virginia {Vassilevska Williams}.
\newblock Popular conjectures imply strong lower bounds for dynamic problems.
\newblock In {\em 55th {IEEE} Annual Symposium on Foundations of Computer
  Science, {FOCS} 2014, Philadelphia, PA, USA, October 18-21, 2014}, pages
  434--443, 2014.

\bibitem{AgarwalR18}
Udit Agarwal and Vijaya Ramachandran.
\newblock Fine-grained complexity for sparse graphs.
\newblock In {\em Proceedings of the 50th Annual {ACM} {SIGACT} Symposium on
  Theory of Computing, {STOC} 2018, Los Angeles, CA, USA, June 25-29, 2018},
  pages 239--252, 2018.

\bibitem{Alon97}
Noga Alon, Raphael Yuster, and Uri Zwick.
\newblock Finding and counting given length cycles.
\newblock {\em Algorithmica}, 17(3):209--223, Mar 1997.

\bibitem{ancona2019}
Bertie Ancona, Monika Henzinger, Liam Roditty, Virginia~Vassilevska Williams,
  and Nicole Wein.
\newblock {Algorithms and Hardness for Diameter in Dynamic Graphs}.
\newblock In Christel Baier, Ioannis Chatzigiannakis, Paola Flocchini, and
  Stefano Leonardi, editors, {\em 46th International Colloquium on Automata,
  Languages, and Programming (ICALP 2019)}, volume 132 of {\em Leibniz
  International Proceedings in Informatics (LIPIcs)}, pages 13:1--13:14,
  Dagstuhl, Germany, 2019. Schloss Dagstuhl--Leibniz-Zentrum fuer Informatik.

\bibitem{Bender00}
Michael~A. Bender and Mart{\'i}n Farach-Colton.
\newblock The {LCA} problem revisited.
\newblock In {\em LATIN 2000: Theoretical Informatics}, pages 88--94, Berlin,
  Heidelberg, 2000. Springer Berlin Heidelberg.

\bibitem{Bjorklund14}
Andreas Bj{\"{o}}rklund, Rasmus Pagh, Virginia {Vassilevska Williams}, and Uri
  Zwick.
\newblock Listing triangles.
\newblock In {\em Automata, Languages, and Programming - 41st International
  Colloquium, {ICALP} 2014, Copenhagen, Denmark, July 8-11, 2014, Proceedings,
  Part {I}}, pages 223--234, 2014.

\bibitem{Bringmann19}
Karl Bringmann, Marvin K\"unnemann, and Karol Wegrzycki.
\newblock Approximating {APSP} without scaling: Equivalence of approximate
  min-plus and exact min-max.
\newblock In {\em Proc.\ STOC}, pages 943--954. ACM, 2019.

\bibitem{Brodal11}
Gerth~Stølting Brodal, Beat Gfeller, Allan~Grønlund Jørgensen, and Peter
  Sanders.
\newblock Towards optimal range medians.
\newblock {\em Theoretical Computer Science}, 412(24):2588--2601, 2011.
\newblock Selected Papers from 36th International Colloquium on Automata,
  Languages and Programming (ICALP 2009).

\bibitem{Driscoll89}
James~R. Driscoll, Neil Sarnak, Daniel~D. Sleator, and Robert~E. Tarjan.
\newblock Making data structures persistent.
\newblock {\em J. Comput. Syst. Sci.}, 38(1):86--124, February 1989.

\bibitem{Duan09}
Ran Duan and Seth Pettie.
\newblock Fast algorithms for (max, min)-matrix multiplication and bottleneck
  shortest paths.
\newblock In {\em Proceedings of the Twentieth Annual {ACM-SIAM} Symposium on
  Discrete Algorithms, {SODA} 2009, New York, NY, USA, January 4-6, 2009},
  pages 384--391, 2009.

\bibitem{DudekG19}
Bart\l{}omiej Dudek and Pawe\l{} Gawrychowski.
\newblock Computing quartet distance is equivalent to counting 4-cycles.
\newblock In {\em Proceedings of the 51st Annual {ACM} {SIGACT} Symposium on
  Theory of Computing, {STOC} 2019, Phoenix, AZ, USA, June 23-26, 2019.}, pages
  733--743, 2019.

\bibitem{Gabow84}
Harold~N. Gabow, Jon~Louis Bentley, and Robert~E. Tarjan.
\newblock Scaling and related techniques for geometry problems.
\newblock In {\em Proceedings of the Sixteenth Annual ACM Symposium on Theory
  of Computing}, STOC '84, pages 135--143, New York, NY, USA, 1984. ACM.

\bibitem{Gold17}
Omer Gold and Micha Sharir.
\newblock {Dominance Product and High-Dimensional Closest Pair under
  $L_\infty$}.
\newblock In Yoshio Okamoto and Takeshi Tokuyama, editors, {\em 28th
  International Symposium on Algorithms and Computation (ISAAC 2017)},
  volume~92 of {\em Leibniz International Proceedings in Informatics (LIPIcs)},
  pages 39:1--39:12, Dagstuhl, Germany, 2017. Schloss Dagstuhl--Leibniz-Zentrum
  fuer Informatik.

\bibitem{Labib19arXiv}
Daniel Graf, Karim Labib, and Przemys\l{}aw Uzna\'nski.
\newblock Hamming distance completeness and sparse matrix multiplication.
\newblock {\em CoRR}, abs/1711.03887, 2017.

\bibitem{itairodeh}
Alon Itai and Michael Rodeh.
\newblock Finding a minimum circuit in a graph.
\newblock {\em SIAM J. Computing}, 7(4):413--423, 1978.

\bibitem{Kent05}
Carmel Kent, Gad~M. Landau, and Michal Ziv-Ukelson.
\newblock On the complexity of sparse exon assembly.
\newblock In {\em Combinatorial Pattern Matching}, pages 201--218, Berlin,
  Heidelberg, 2005. Springer Berlin Heidelberg.

\bibitem{Kopelowitz16}
Tsvi Kopelowitz, Seth Pettie, and Ely Porat.
\newblock Higher lower bounds from the {3SUM} conjecture.
\newblock In {\em Proceedings of the Twenty-seventh Annual ACM-SIAM Symposium
  on Discrete Algorithms}, SODA '16, pages 1272--1287, Philadelphia, PA, USA,
  2016. Society for Industrial and Applied Mathematics.

\bibitem{Labib19}
Karim Labib, Przemys\l{}aw Uzna\'nski, and Daniel Wolleb-Graf.
\newblock {Hamming Distance Completeness}.
\newblock In Nadia Pisanti and Solon~P. Pissis, editors, {\em 30th Annual
  Symposium on Combinatorial Pattern Matching (CPM 2019)}, volume 128 of {\em
  Leibniz International Proceedings in Informatics (LIPIcs)}, pages
  14:1--14:17, Dagstuhl, Germany, 2019. Schloss Dagstuhl--Leibniz-Zentrum fuer
  Informatik.

\bibitem{LeGall14}
Fran\c{c}ois Le~Gall.
\newblock Powers of tensors and fast matrix multiplication.
\newblock In {\em Proceedings of the 39th International Symposium on Symbolic
  and Algebraic Computation}, ISSAC '14, pages 296--303, New York, NY, USA,
  2014. ACM.

\bibitem{LeGall18}
Fran\c{c}ois Le~Gall and Florent Urrutia.
\newblock Improved rectangular matrix multiplication using powers of the
  {Coppersmith-Winograd} tensor.
\newblock In {\em Proceedings of the Twenty-Ninth Annual ACM-SIAM Symposium on
  Discrete Algorithms}, SODA '18, pages 1029--1046, Philadelphia, PA, USA,
  2018. Society for Industrial and Applied Mathematics.

\bibitem{lincolnsoda18}
Andrea Lincoln, Virginia {Vassilevska Williams}, and R.~Ryan Williams.
\newblock Tight hardness for shortest cycles and paths in sparse graphs.
\newblock In {\em Proceedings of the Twenty-Ninth Annual {ACM-SIAM} Symposium
  on Discrete Algorithms, {SODA} 2018, New Orleans, LA, USA, January 7-10,
  2018}, pages 1236--1252, 2018.

\bibitem{Matousek91}
Jiří Matoušek.
\newblock Computing dominances in {$E^n$}.
\newblock {\em Inf. Process. Lett.}, 38(5):277--278, 1991.

\bibitem{nesetrilpoljak}
Jaroslav Ne\v{s}et\v{r}il and Svatopluk Poljak.
\newblock On the complexity of the subgraph problem.
\newblock {\em Commentationes Math. Universitatis Carolinae}, 26(2):415--419,
  1985.

\bibitem{Okasaki98}
Chris Okasaki.
\newblock {\em Purely Functional Data Structures}.
\newblock Cambridge University Press, New York, NY, USA, 1998.

\bibitem{Patrascu10}
Mihai P\v{a}tra\c{s}cu.
\newblock Towards polynomial lower bounds for dynamic problems.
\newblock In {\em Proceedings of the Forty-second ACM Symposium on Theory of
  Computing}, STOC '10, pages 603--610, New York, NY, USA, 2010. ACM.

\bibitem{RodittyW11}
Liam Roditty and Virginia {Vassilevska Williams}.
\newblock Minimum weight cycles and triangles: Equivalences and algorithms.
\newblock In {\em {IEEE} 52nd Annual Symposium on Foundations of Computer
  Science, {FOCS} 2011, Palm Springs, CA, USA, October 22-25, 2011}, pages
  180--189, 2011.

\bibitem{Schank05}
Thomas Schank and Dorothea Wagner.
\newblock Finding, counting and listing all triangles in large graphs, an
  experimental study.
\newblock In {\em Experimental and Efficient Algorithms}, pages 606--609,
  Berlin, Heidelberg, 2005. Springer Berlin Heidelberg.

\bibitem{VWAPBP}
Virginia Vassilevska, Ryan Williams, and Raphael Yuster.
\newblock All-pairs bottleneck paths for general graphs in truly sub-cubic
  time.
\newblock In {\em Proceedings of the 39th Annual {ACM} Symposium on Theory of
  Computing, San Diego, California, USA, June 11-13, 2007}, pages 585--589,
  2007.

\bibitem{VWAPBPj}
Virginia Vassilevska, Ryan Williams, and Raphael Yuster.
\newblock All pairs bottleneck paths and max-min matrix products in truly
  subcubic time.
\newblock {\em Theory of Computing}, 5(1):173--189, 2009.

\bibitem{vstoc12}
Virginia {Vassilevska Williams}.
\newblock Multiplying matrices faster than {C}oppersmith-{W}inograd.
\newblock In {\em Proceedings of the 44th Symposium on Theory of Computing
  Conference, {STOC} 2012, New York, NY, USA, May 19 - 22, 2012}, pages
  887--898, 2012.

\bibitem{cs367hw}
Virginia {Vassilevska Williams}.
\newblock {Homework 2 of CS367 at Stanford from 2015 (and 2014)}.
\newblock \url{http://theory.stanford.edu/~virgi/cs367/hw2.pdf}, 2015.

\bibitem{focsy}
Virginia~Vassilevska Williams and Ryan Williams.
\newblock Subcubic equivalences between path, matrix and triangle problems.
\newblock In {\em 51th Annual {IEEE} Symposium on Foundations of Computer
  Science, {FOCS} 2010, October 23-26, 2010, Las Vegas, Nevada, {USA}}, pages
  645--654, 2010.

\end{thebibliography}

\end{document}